\documentclass[letterpaper,11pt,onecolumn,unpublished]{quantumarticle} 
\pdfoutput=1
\usepackage[utf8]{inputenc}
\usepackage[english]{babel}
\usepackage[T1]{fontenc}
\usepackage{hyperref}

\usepackage{amsmath,amsthm,amsfonts,amssymb,braket}
\usepackage[all]{xy}
\usepackage{fullpage}
\usepackage{mathtools}
\usepackage[numbers]{natbib}

\usepackage[nameinlink,capitalize]{cleveref}

\newtheorem{lemma}{Lemma}[section]
\newtheorem{theorem}[lemma]{Theorem}

\newtheorem{corollary}[lemma]{Corollary}

\newtheorem{proposition}[lemma]{Proposition}
\newtheorem{proposition-definition}[lemma]{Proposition-Definition}
\theoremstyle{definition}
\newtheorem{definition}[lemma]{Definition}
\newtheorem{remark}[lemma]{Remark}
\newtheorem{example}[lemma]{Example}
\newtheorem{problem}[lemma]{Problem}

\newcommand{\CC}{{\mathbb{C}}}

\newcommand{\ZZ}{{\mathbb{Z}}}

\newcommand{\FF}{{\mathbb{F}}}

\newcommand{\calB}{{\mathcal{B}}}

\newcommand{\calD}{{\mathcal{D}}}
\newcommand{\calE}{{\mathcal{E}}}

\newcommand{\calM}{{\mathcal{M}}}

\newcommand{\one}{{\mathbf{1}}}

\usepackage{tikz}
\usepackage{lipsum}
\usepackage{subcaption}

\newcommand{\ISG}{{\mathsf{ISG}}}
\newcommand{\barISG}{{\overline{\mathsf{ISG}}}}

\DeclareMathOperator{\rank}{{rank}}

\DeclareMathOperator{\ance}{{\mathbf{ance}}}

\newcommand{\proj}[1]{\ket{#1}\!\bra{#1}}

\begin{document}

\title{The code distance of Floquet codes}

\author{Keller Blackwell}
\affiliation{Department of Computer Science, Stanford University}

\author{Jeongwan Haah}
\affiliation{Leinweber Institute for Theoretical Physics, Stanford University}
\affiliation{Google Quantum AI}
\orcid{0000-0002-1087-6853}
\date{}

\maketitle 

\begin{abstract}
    For fault-tolerant quantum memory defined by periodic Pauli measurements,
    called \emph{Floquet codes},
    we prove that every correctable, undetectable spacetime error
    occurring during the steady stage
    is a product of 
    (i) measurement operators inserted at the time of the measurement
    and (ii) pairs of identical Pauli operators sandwiching a measurement that commutes with
    the operator.
    We call such errors \emph{benign};
    they define a binary vector subspace of spacetime errors
    which properly generalize stabilizers of static Pauli stabilizer codes.
    Hence, the code distance of a Floquet code
    is the minimal weight of an undetectable spacetime Pauli error that is not benign.
    Our results apply more generally to families of dynamical codes for which
    every instantaneous stabilizer is inferred from measurements
    in a time interval of bounded length.
\end{abstract}

\section{Introduction}

In a conventional construction of fault-tolerant quantum memory, 
one starts with a Pauli stabilizer code 
before designing a Clifford circuit to measure the stabilizers.
The latter is as important as the former,
and it is well known that the error correction capability of a Pauli stabilizer code
as indicated by its code distance is only a performance target,
\emph{not} a guarantee.
The appreciation of circuit level noise has grown over time,
and realizable quantum error correction must be centered around circuits
rather than a subspace of the Hilbert space of a many-qubit system~\cite{Gottesman2022,McEwen2023,Delfosse2023}.

Prior work~\cite{Hastings2021} establishes
that carefully constructed measurement dynamics can function as a fault-tolerant quantum memory
despite not stabilizing a subspace sufficient to encode even a single logical qubit.
Perhaps the simplest class of such quantum memory 
is defined by periodic sequences of Pauli measurements, called \emph{Floquet codes}.
Of course, there is no fundamental reason to require measurement periodicity,
and one can also consider more general \emph{dynamical codes}~\cite{Gottesman2022}
defined by Clifford circuits with Pauli measurements.

The code distance of Floquet or dynamical codes
must be the minimal weight of an undetected spacetime error that causes a logical failure.
While the detectability of such an error is relatively straightforward~\cite{McEwen2023,Delfosse2023} 
(but see also~\cite{Fu2024}),
the question of whether the error causes a logical failure
is often subtle, being highly sensitive to time boundary conditions.

In some concrete examples of Floquet codes
({\it e.g.}, the honeycomb code without~\cite{Hastings2021,Fahimniya2023} 
and with~\cite{Haah2022,Paetznick2022,Gidney2022} boundaries,
the CSS honeycomb code~\cite{Davydova2022}, 
the X-cube Floquet code~\cite{Zhang2022},
the coupled spin chain construction~\cite{Yan2024},
and Floquet Bacon--Shor code~\cite{Alam2024}),
the code distance is estimated by analyzing the map from errors to sydromes~\cite{McEwen2023}.
These analyses happen to rely on decoding graphs,
meaning that there are exactly two flipped syndrome bits for each elementary error,
which is not always true in general Floquet codes.

It is not always straightforward 
to identify logical failures caused by spacetime error configurations.
\cite{Vuillot2021} constructs a Floquet code on the honeycomb graph for which
all instantaneous stabilizer groups have $O(\sqrt{n})$ code distance,
but whose time dynamics allow certain $O(1)$ weight spacetime errors to induce logical failures
while remaining undetected.
Moreover,
some errors may force the physical state out of the instantaneous code space;
nonetheless, under certain measurement dynamics, subsequent measurements 
can be not only unaffected --- they may even \emph{absorb} the error.
Prior works have called this phenomenon \emph{self-correction}~\cite{Alam2024}.

This subtlety is rooted in the status
that there is no criterion to tell whether an undetectable spacetime error
will eventually have implemented a nontrivial logical operation on the encoded qubit.
Such uncertainty is in contrast with conventional Pauli stabilizer codes,
where any undetectable error must fall in precisely one of two classes:
the trivial logical class consisting of all stabilizers,
and the non-trivial logical class consisting of logical operators.
The code distance of a Pauli stabilizer code is therefore
the minimum weight of an undetectable error that is not a stabilizer.
Classifying a given undetectable error within this dichotomy is efficient 
and amounts to checking its commutativity with any basis of the logical operators.

Here, we give such a dichotomy for undetectable spacetime errors in Floquet codes.
We identify a class of errors which manifestly preserve logical states,
generated by two sets of errors:
(i) a Pauli measurement operator inserted at the time of the measurement,
and (ii) a pair of identical Pauli operators
\footnote{One is the inverse of the other in case of prime qudit codes.}
inserted immediately before and after a measurement
which commutes with the operator.
We prove that any undetectable spacetime error inserted during the steady stage of a Floquet code
that is not generated by the aforementioned errors 
implements a nontrivial logical operation on the encoded qubits.
This conclusion follows from the result that, in any Floquet code,
every undetectable spacetime error is equivalent to a
logical operator of an instantaneous Pauli stabilizer code at some future time.
We show that, given an undetectable spacetime error, computation of 
the equivalent logical operator is efficient; hence, as in the stabilizer code case,
such a dichotomy can be checked efficiently.
Our results hold in any dynamical code (and in particular, Floquet codes)
for which every instantaneous stabilizer is inferred by measurements 
from a constant-width time window, which we call the 
\emph{bounded-inference property}.

\subsection{Prior work}

Delfosse and Paetznick~\cite{Delfosse2023}
considered arbitrary, finite Clifford circuits that include Pauli measurements.
They defined two closely related codes:
the \emph{outcome code} is a binary vector space (a classical code) consisting of all 
possible outcomes of the measurements in a circuit
in the absence of noise,
and the \emph{spacetime code} is a Pauli stabilizer code which maps 
spacetime Pauli errors to deviations of the measurement outcomes 
from the outcome code. 
In other words, their work focuses on \emph{detection} of spacetime errors.
This would provide sufficient machinery
to study finite Clifford circuits whose operational meaning is self-contained,
since one could encode the final logical outcome into some Pauli measurements.
However, it is not always clear how to adapt this machinery 
if the finite circuit is treated as a fault-tolerant component ({\it e.g.}, a quantum memory)
which maintains its own performance guarantees in modular applications.
For Floquet codes, 
one might consider a sequence of Clifford circuits,
starting with some logical encoding and ending with logical measurements
with increasing memory times. The problem remains as to how to
ascribe the error correction performance of the resulting circuit
to the time boundary conditions or the underlying Floquet code;
see also~\cref{rem:DP2023}.

Fu and Gottesman~\cite{Fu2024} study dynamical codes specified 
by an initial Pauli stabilizer stabilizer code $S_0$ and a finite 
sequence of Pauli measurements $\mathbf{M}$. A code state 
$\mathbf{c}$ of $S_0$ is prepared, after which $\mathbf{M}$ is 
measured through \emph{once}; their fault model considers a 
one-time error inserted immediately following the preparation of 
$\mathbf{c}$. Our setting is considerably different: we do not fix 
$S_0$ a priori, instead considering a periodic measurement 
schedule of length $P$. The instantaneous stabilizer codes are 
induced entirely by the recurrent measurement dynamics, which 
persist indefinitely. Most importantly, whereas~\cite{Fu2024} 
considers a one-time error insertion, we consider general 
spacetime errors which may be inserted anywhere in spacetime 
after the code enters \emph{steady stage} (when the 
instantaneous codes reaches a constant dimension).

There exists a graphical representation of quantum circuits, 
called ZX calculus, 
to study Clifford circuits with Pauli measurements~\cite{Wetering2020}.
Bomb\'in {\it et al.}~\cite{Bombin2024} used this approach on several examples, illustrating that detectors and logical failures can be probed
using the commutation relations between errors and ``Pauli webs''
corresponding to detectors and logical operators.
\cite{Bombin2024} focused on quantum circuits measuring Pauli observables
supported in $X, Z$.
Magdalena de la Fuente {\it et al.}~\cite{Fuente2024}
extended ZX calculus to a tri-color ZXY calculus which
more naturally characterizes circuits with arbitrary Pauli measurements,
and showed how Pauli webs can be constructed in this more
general setting.
Xu and Dua~\cite{Xu2025} applied this approach to
dynamical codes obeying certain spacetime-locality conditions, arguing 
for their fault-tolerance.
In both~\cite{Fuente2024} and~\cite{Xu2025},
it is important that there are input and output Pauli stabilizer codes
that are finitely separated in time,
and the same limitation as~\cite{Delfosse2023} applies.
See also~\cref{rem:LogicalFailureProbeExtensionQuestion}.

Given the subtlety of probing logical failures without future time boundaries,
one may wonder how numerical study of Floquet codes is even possible.
This can indeed be subtle in general 
as one must often impose temporal boundaries to make a simulated circuit finite,
and the temporal boundary may need to be tailored to 
the exact construction under consideration.
On topological codes, the issue is less severe
since the correlation between error correction operators
is short ranged in time~\cite{Dennis2001}.
Typically,
error-free time steps and noiseless measurements are inserted 
in last stages of the circuit~\cite{Gidney2021,Gidney2022,Paetznick2022,McEwen2023},
and simulation software directly gives what remains after decoding.
A justification for these noiseless final rounds
is that any logical qubit will be measured out destructively in practice,
and all these simulated quantum codes have a property
that single-qubit measurements allow for reconstruction of 
all logical and syndrome bits,
which is a property enjoyed by all CSS Pauli stabilizer codes.
It however remains true that these time boundary conditions
make it difficult to compare Floquet codes.

\subsection{Settings}

We recall the evolution of the stabilizer group under a Pauli measurement schedule. 
Suppose that the system is the maximally mixed state in a common eigenspace of
a Pauli stabilizer group, called an {\bf instantaneous stabilizer group} or $\ISG(t)$ at time~$t$,
and that we measure a Pauli operator~$M$ in the subsequent time step.
If $M \in \ISG(t)$, then the physical state is an eigenstate of~$M$.
The measurement reveals this eigenvalue, and the physical state is undisturbed;
hence, $\ISG(t+1) = \ISG(t)$.
If on the other hand $M \notin \ISG(t)$, we must consider two cases:
\begin{enumerate}
    \item[(i)] $M$ commutes with every element of~$\ISG(t)$, or
    \item[(ii)] $M$ fails commute with some element $S \in \ISG(t)$.
\end{enumerate}
In case of~(i),
$M$ is a nontrivial logical operator of the instantaneous code,
and the post-measurement physical state is projected to an eigenstate of~$M$.
Depending on the outcome, $M$ or $-M$ becomes a new stabilizer,
and $\ISG(t+1) = \langle \pm M, \ISG(t) \rangle$;
the rank of the stabilizer group increases by~$1$.
In case of~(ii),
since the post-measurement physical state must be an eigenstate of~$M$,
the anticommuting operator~$S$ cannot be in~$\ISG(t+1)$.
On the other hand, any element of~$\ISG(t)$ that commutes with~$M$
is still a stabilizer.
The new stabilizer group~$\ISG(t+1)$ 
is generated by this $M$-commutant subgroup of~$\ISG(t)$
and $\pm M$;
the rank of the stabilizer group stays the same.

\begin{definition}\label{defn:dynamicalCode}
    A {\bf dynamical code} is defined by a time sequence~$\mathbf M_1, \mathbf M_2, \ldots$ 
    of sets $\mathbf M_t = \{M_{t,i}\}_i$ of Pauli operators~$M_{t,i}$,
    where $M_{t,i}$ and $M_{t,j}$ must commute for all~$i,j$.
    There is no requirement that operators $M_{t,i}$ and $M_{t',j}$ 
    at different time steps~$t \neq t'$ need to commute.
    If the sequence $\mathbf M_t$ is periodic in~$t$,
    then we call the dynamical code a {\bf Floquet code}.
\end{definition}

\noindent 
The implementation of a dynamical code is nothing but measuring $M_{t,i}$ for all~$i$ at time~$t$.
Since all the measurements at one time step is assumed to commute with each other,
they can be implemented simultaneously.

\begin{remark}\label{rem:CliffordUnitary}
    One can include Clifford unitaries in the definition of a dynamical or Floquet code,
    but we do not consider them in this paper for simplicity.
    Note that all our results can be adapted in the unitary-allowed settings
    by conjugating Pauli errors by the unitary gates in the code.
\end{remark}

Our dynamics always start with the maximally mixed state
and the trivial instantaneous stabilizer group~$\ISG(0) = \{\one\}$.
Later instantaneous stabilizer groups~$\ISG(t)$ are not determined 
before one actually implements the schedule.
However, it is standard to calculate 
which measurements will have deterministic or nondeterministic outcomes,
and, in this paper, $\ISG(t)$ will mean an instantaneous stabilizer group
determined by any valid set of measurement outcomes 
assuming that no errors occur during the dynamics. 

Although the measured signs of Pauli operators are central to 
\emph{performing} Pauli error correction,
much of our derivation will not require detailed sign information.
Thus, we will use 
\begin{equation}
    \barISG = \langle -\one, \ISG \rangle
\end{equation}
to denote the extension of the instantaneous stabilizer group by signs.
This eases notation;
for example, $M \in \barISG$ means that either $M \in \ISG$ or $-M \in \ISG$.
Additionally, we will refer to ``logical operators of $\ISG$'' as shorthand for the more precise phrase ``logical operators of the Pauli stabilizer code defined by~$\ISG$.''

\begin{definition}
    Let $\mathcal P$ be the Pauli group on~$n$ qubits,
    and let $\mathcal P_t$ for any integer $t \ge 0$ be a copy of $\mathcal P$.
    A {\bf spacetime Pauli error} $E = (\ldots, E_t, \ldots)$, 
    or just an {\bf error} for short,
    is an element of the infinite direct sum~$\calE = \bigoplus_{t \ge 0} \mathcal P_t$.
    If $I$ is a set of time steps such that $E_t \in \CC\one$ for all $t \notin I$,
    then we say $E$ is {\bf supported} on~$I$.
    We write $[O]_{t'}$ for any Pauli operator~$O$
    to mean the Pauli operator~$O$ inserted at time~$t'$;
    that is, the bracket~$[\cdot]$ strips off the time coordinate, if any,
    and resets it by the subscript.
\end{definition}

\noindent
If we mod out all phase factors, the set~$\calE$ becomes a $\FF_2$-linear space
where the addition of vectors $E = (\ldots,\, E_t, \, E_{t+1},\, \ldots)$
and $E' = (\ldots,\, E_t', \, E_{t+1}',\, \ldots)$ is inherited from an obvious multiplication rule of operators:
$EE' = (\ldots,\, E_t E_t',\, E_{t+1} E_{t+1}',\, \ldots)$.

\begin{definition}[Error timing convention]
    Consider a dynamical code with measurement schedule $\mathbf{M}_1, \mathbf{M}_2, \ldots$ 
    where $\mathbf{M}_t$ denotes the set of measurements performed at time~$t$. 
    For any error~$E = (\ldots, E_t, \ldots) \in \calE$, 
    the error component~$E_t \in\mathcal{P}_t^n$ is inserted immediately 
    \emph{following} all measurements in~$\mathbf{M}_t$ have been performed.
\end{definition}

\subsection{Action of errors}

We can regard any measurement as a quantum channel 
that projects the system and adjoins measurement outcomes.
For clarity, let us momentarily assume (without loss of generality) that there is one measurement at a time.
Then, the measurement channel will be
\begin{equation}
    \calM_{t+1} : \rho_t \mapsto \rho_{t+1} 
    = 
    \Pi \rho_t \Pi \otimes \proj 0 
    + (\one - \Pi) \rho_t (\one - \Pi) \otimes \proj 1
\end{equation}
If a measurement is deterministic, one of the two terms is zero;
otherwise, they both have trace~$1/2$ in the steady stage.
The measurement result register gets enlarged by one bit every measurement.
At any given execution of a dynamical code,
we will have a pure state on the measurement outcome register,
but here we consider the probabilistic ensemble over all measurement outcomes.
A spacetime error~$E$ gives additional insertions of Pauli channels 
in the sequence of measurement channels,
and modifies the chain of outcome-recorded density matrices:
\begin{equation}
    E : \begin{pmatrix}
        \rho_0\\ 
        \rho_1 = \calM_1(\rho_0)\\
        \vdots\\
        \rho_t = \calM_t(\rho_{t-1})\\
        \vdots
    \end{pmatrix} 
    \mapsto 
    \begin{pmatrix}
        \rho'_0 = E_0 \rho_0 E_0^\dagger\\
        \rho'_1 = E_1 \calM_1(E_0 \rho_0 E_0^\dagger) E_1^\dagger\\
        \vdots\\
        \rho'_t = E_t \calM_t(\rho'_{t-1}) E_t^\dagger\\
        \vdots
    \end{pmatrix}
\end{equation}

\begin{proposition}\label{thm:ErrorAction}
    We have a group action of the group~$\calE$ of all spacetime errors
    on the set of all chains of outcome-recorded density matrices.
\end{proposition}

\begin{proof}
    Clearly, the empty error acts by the identity.
    We have to show the associativity.
    Let $E,F \in \calE$.
    Applying~$E$ and then~$F$, we have
    \begin{equation}
        \begin{pmatrix}
            \vdots\\
            \rho_t = \calM_t(\rho_{t-1})\\
            \vdots
        \end{pmatrix} 
        \mapsto 
        \begin{pmatrix}
            \vdots\\
            \rho'_t = E_t \calM_t(\rho'_{t-1}) E_t^\dagger\\
            \vdots
        \end{pmatrix}
        \mapsto
        \begin{pmatrix}
            \vdots\\
            \rho''_t = F_t E_t \calM_t(\rho''_{t-1}) E_t^\dagger F_t^\dagger\\
            \vdots
        \end{pmatrix} \, ,
    \end{equation}
    which is the same as applying $FE \in \calE$ at once:
    \begin{equation}
        \begin{pmatrix}
            \vdots\\
            \rho_t = \calM_t(\rho_{t-1})\\
            \vdots
        \end{pmatrix} 
        \mapsto 
        \begin{pmatrix}
            \vdots\\
            \rho''_t = F_t E_t \calM_t(\rho''_{t-1}) E_t^\dagger F_t^\dagger\\
            \vdots
        \end{pmatrix} \, .\qedhere
    \end{equation}
\end{proof}

\section{Ancestries of stabilizers and detectors}

In this section, 
we will understand how error detection works in dynamical codes.
This has been discussed previously in various forms~\cite{McEwen2023,Bombin2024},
but our exposition clarifies how instantaneous stabilizer groups are affected by errors.
We will first examine how a current stabilizer 
is constructed from past measurements in the absence of any errors,
and then explain why certain errors leave detectable signatures in measurement outcomes.
By considering reference systems, this knowledge will prove useful to understand logical failures.
We will identify for each detector an element of~$\calE$
whose commutation relation with a given error determines whether the detector will be triggered.

To avoid unnecessary complication,
we assume in this section that exactly one measurement is taken at a time 
even if many measurements could be implemented simultaneously.
Under this assumption, every element~$S \in \ISG(t)$ 
is derived from a unique ancestry,
defined recursively as follows.

\begin{definition}\label{defn:ancestry}
    Let $0 \le \tau \le t$ be time steps,
    $S_t \in \ISG(t)$ a stabilizer,
    and $M_t$ the measurement operator at time~$t$ with outcome~$m_t \in \{\pm 1\}$
    so that $m_t M_t \in \ISG(t)$.
    We define the {\bf ancestry} of~$S_t$, denoted $\ance_\tau^t (S_t) \in \calE$, by
    \begin{equation}
        \ance_\tau^t (S_t) = 
        \begin{cases}
            \one_t \cdot \ance^{t-1}_\tau([S_{t}]_{t-1}) & \text{if } [S_t]_{t-1} \in \ISG(t-1) \text{ and } t > \tau,\\
            (m_t M_t) \cdot \ance^{t-1}_\tau([(m_t M_t)^\dagger S_t ]_{t-1}) & \text{if } [S_t]_{t-1} \notin \ISG(t-1) \text{ and } t > \tau,\\
            S_t & \text{if } t = \tau.
        \end{cases} \label{eq:ancestry}
    \end{equation}
\end{definition}

\begin{lemma}\label{thm:ancestryWellDefined}
    In the second case of~\eqref{eq:ancestry}, 
    we always have $[ (m_t M_t)^\dagger S_t]_{t-1} \in \ISG(t-1)$ which commutes with~$[M_t]_{t-1}$.
    Therefore, the ancestry is well-defined.
\end{lemma}

\begin{proof}
    The condition that $[S_t]_{t-1} \notin \ISG(t-1)$ 
    implies that $S_t$ is a new member introduced to~$\ISG(t)$
    because of the measurement of~$M_t$.
    (The measurement of~$M_t$ cannot just change the sign of~$S_t$,
    and hence $-[S_t]_{t-1} \notin \ISG(t-1)$.)
    That is, $S_t \in (m_t M_t) \cdot \mathcal S \subseteq (m_t M_t) \cdot [\ISG(t-1)]_t$
    where $\mathcal S \subseteq [\ISG(t-1)]_t$ consists of all those that commute with~$M_t$.
\end{proof}

The ancestry of a stabilizer tells us which prior measurement outcomes 
are used to infer the eigenvalue of a present stabilizer.
The recursive definition gives an algorithm to find the ISG elements in the past
upon which the eigenvalue of a current stabilizer depends.

\begin{corollary}\label{thm:StabilizerFromAncestry}
    Let $S_t \in \ISG(t)$ and
    let $A_t A_{t-1} \cdots A_{\tau +1} B_\tau = \ance_\tau^t(S_t)$
    with $B_\tau \in \ISG(\tau)$ 
    where
    $A_j \neq \one_j$ only if $A_j = m_j M_j \in \ISG(j)$
    (with the outcome~$m_j = \pm 1$ of~$M_j$).
    Define for each $j \in [\tau+1, t]$
    \begin{equation}
        P_j = [A_{j}]_j [A_{j-1}]_j \cdots [A_{\tau+1}]_j [B_\tau]_j  \in \ISG(j) \, .
    \end{equation}
    Then, for each $j \in [\tau+1, t]$, the operator~$A_j$ commutes with~$P_j$ and 
    $[A_j^\dagger P_j]_{j-1} \in \ISG(j-1)$.
    We also have that $S_t = P_t$.
\end{corollary}

\begin{proof}
    The second claim that~$S_t = P_t$ is clear from the definition of the ancestry.
    If $A_j = \one_j$, the first claim is trivial,
    which appears in the first case of~\eqref{eq:ancestry}.
    If $A_j$ is the measurement operator at time~$j$,
    the second case of~\eqref{eq:ancestry} must have happened.
    By~\cref{thm:ancestryWellDefined}, we know that 
    $[A_j^\dagger P_j]_{j-1} \in \ISG(j-1)$ and it must commute with~$[A_j]_{j-1}$.
\end{proof}

\begin{example}
    Suppose that we measure $XI$ at time~$1$, $IX$ at time~$2$, $XX$ at time~$3$, and $ZZ$ at time~$4$.
    Then, $YY \in \barISG(4)$ has the ancestry $[ZZ]_4 [II]_3 [IX]_2 [XI]_1$ up to signs.
\end{example}

From~\cref{thm:StabilizerFromAncestry}
we have an expression for any current stabilizer~$S_t$
as a product of measurement operators,
which is correct including all the signs.
Note that every nonidentity operator~$A_j$ for $j \in [\tau+1,t]$
is associated with a \emph{non}deterministic measurement;
if it were deterministic, the definition of ancestry must have skipped it.
Hence, there cannot be any information about errors in the outcomes of~$M_j$.
The only useful information is obtained when~$S_t$ gets measured in a subsequent step.

Suppose an error~$E_\tau$ is inserted in one time step~$\tau$.
This is after $M_\tau$ is measured, but before $M_{\tau+1}$ is measured.
Before the error, the physical state is stabilized by~$B_\tau$,
but after the error it is stabilized by~$\delta B_\tau$
where $\delta = 1$ if $E_\tau$ commutes with~$B_\tau$ and $\delta = -1$ otherwise.
At each subsequent time step~$j$, we have
\begin{equation}
    P'_j = [m_j M_j]_j [m_{j-1} M_{j-1}]_{j} \cdots [m_{\tau+1} M_{\tau+1}]_j [\delta B_\tau]_j \in \ISG'(j)
\end{equation}
where $m_j$ is the outcome of the measurement by~$M_j$.
This differs by~$\delta$ from what is predicted from measurement outcomes:
\begin{equation}
    P_j = [m_j M_j]_j [m_{j-1} M_{j-1}]_{j-1} \cdots [m_{\tau+1} M_{\tau+1}]_j [B_\tau]_j \in \ISG(j).
\end{equation}
Suppose that an additional one-time-slice error~$F_{\tau'}$ 
is inserted at time~$\tau' \ge \tau$.
Before $F_{\tau'}$ is inserted but after $M_{\tau'}$ is measured,
the physical state is stabilized by~$P'_{\tau'} = \delta P_{\tau'}$.
Inserting~$F_{\tau'}$, the physical state is stabilized by~$\eta P'_{\tau'} = \eta \delta P_{\tau'}$
where $\eta = \pm 1$ is determined by the commutation relation 
between~$P_{\tau'}$ and $F_{\tau'}$.
Arriving at time $t \ge \tau' \ge \tau$,
we see that the state with $F_{\tau'} E_\tau$ inserted
is stabilized by~$\eta \delta S_t$.
Observe that $P_{\tau}$ that was responsible for the commutation relation
is the time-$\tau$ component of~$\ance_\tau^t(S_t)$.
The same is true for~$P_{\tau'}$ with $\tau'$ replacing~$\tau$.

More generally, we arrive at
\begin{proposition}\label{thm:HowDetectorIsTriggered}
    Let $\pi_\tau : (\calE = \bigoplus_t \mathcal P_t) \to \mathcal P_\tau$ 
    be the canonical projection of the direct sum.
    Let $S_t \in \ISG(t)$ be an instantaneous stabilizer 
    and $E = (\ldots,E_j,\ldots)$ be an error supported on a time interval~$[a,b]$.
    The physical state at time~$t \ge b$ after the insertion of~$E$ is stabilized by~$(-1)^\sigma S_t$
    where
    \begin{align}
        \sigma &= \sum_{j = a}^b \lambda\Big(E_j,~ \pi_j \circ \ance_j^t(S_t)\Big) \mod 2 \, , \\
        \lambda(P,Q) &= \begin{cases} 0 & \text{if $P$ and $Q$ commute}, \\ 1 &\text{othewise} \end{cases} \nonumber
    \end{align}
\end{proposition}

Hence, the measurement of an existing stabilizer is a {\bf detector}~\cite{Gidney2021,McEwen2023}.
Conversely, since we make one measurement at a time in this section
(even though many measurements can be implemented simultaneously),
any detector is a measurement at time step~$t$ of an element of~$\ISG(t-1)$.
\footnote{
    In general, a detector is a linear constraint on measurement outcomes at time~$t$
    given all the outcomes from the earlier time steps~\cite{Gidney2021,McEwen2023,Delfosse2023}.
}
We define the {\bf ancestry of a detector} at time~$t$ 
to be the ancestry of the stabilizer at time~$t-1$ that the detector measures.
\Cref{thm:HowDetectorIsTriggered} says that
we can test whether a given detector at time~$t$ is triggered 
by a spacetime error~$E$ 
from the commutation relation between $E$ and
\begin{equation}
    D = \Big( \pi_j \circ \ance^{t-1}_j (\pm [M_{t}]_{t-1}) \Big)_{j=0}^{t-1} \quad \in \calE \, . \label{eq:detectorAsError}
\end{equation}

\begin{definition}
    An {\bf undetectable} error is one that does not trigger any detector any time.
\end{definition}

The syndrome bit from a detector is the \emph{difference} 
between the inferred eigenvalue of a stabilizer
from its ancestry and the current measured value.
\footnote{
    In~\cite{Hastings2021} the calculation for syndrome bits for a given error
    was performed by looking at measurements that are ``flipped'' by an error
    and the criterion for a flipped measurement is whether the measurement operator
    anticommutes with a given error ignoring all the time coordinates.
    While this rule gave correct results,
    it is not entirely sound reasoning to say that 
    ``the measurement outcome is flipped because the error anticommutes with the measurement operator''
    especially if the measurement is nondeterministic regardless of the error.
    After all, any detector probes the flip of some element of~$\ISG(\tau)$.
}
In~\cite{Bombin2024,Fuente2024,Xu2025},
detectors are described by certain ``Pauli webs'' or ``Pauli flows''
that are contained in the interior of a ZX diagram,
which is the same as~$D$ in~\eqref{eq:detectorAsError}.
Our exposition emphasizes operational meaning of~$D$.

For detectors associated with persistent stabilizers
that are used for the syndrome bits of the honeycomb code~\cite{Hastings2021},
the differences of the measured eigenvalues between nearer time steps
are taken as a basis for the syndrome.
More generally, we can understand such a basis as follows.

\begin{proposition}\label{thm:anceIsLinear}
    For any $t \ge \tau \ge 0$,
    the map~$\ance_\tau^t : \ISG(t) \to \calE$ with all signs ignored
    is $\FF_2$-linear.
\end{proposition}

\begin{proof}
    Since we ignore all signs for Pauli operators, we will suppress signs in this proof.
    We use induction in~$t$.
    If $t = \tau$, the map is an inclusion map.
    Assume the claim for $\ance_\tau^{t-1}: \ISG(t-1) \to \calE$.
    Let $A_t, B_t \in \ISG(t)$ and $M_t$ the measurement operator at time~$t$.
    
\begin{itemize}
    \item[(i)] Suppose that $[A_t]_{t-1}, [B_t]_{t-1} \in \ISG(t-1)$.
    Then, it follows that $\ance_\tau^t(A_t) = \ance_\tau^{t-1}([A_t]_{t-1})$ and 
    $\ance_\tau^t(B_t) = \ance_\tau^{t-1}([B_t]_{t-1})$.
    By the induction hypothesis, their product is given by
    $\ance_\tau^{t-1}([A_t B_t]_{t-1})$,
    which equals $\ance_\tau^t(A_tB_t)$ since $[A_t B_t]_{t-1} \in \ISG(t-1)$.
    
    \item[(ii)] Suppose that $[A_t]_{t-1} \in \ISG(t-1)$ and $[B_t]_{t-1} \notin \ISG(t-1)$.
    Then, $[A_t B_t]_{t-1} \notin \ISG(t-1)$.
    The three relevant images of~$\ance_\tau$ are
    $\ance_\tau^t(A_t) = \ance_\tau^{t-1}([A_t]_{t-1})$,
    $\ance_\tau^t(B_t) = M_t \cdot \ance_\tau^{t-1}([M_t^\dagger B_t ]_{t-1})$ 

    $\ance_\tau^t(A_t B_t) = M_t \cdot \ance_\tau^{t-1}([M_t^\dagger A_t B_t ]_{t-1})$,
    which are consistent with linearity.
    
    \item[(iii)] Suppose that $[A_t]_{t-1} \notin \ISG(t-1)$ and $[B_t]_{t-1} \notin \ISG(t-1)$.
    Then, $\ISG(t) = \langle M_t \rangle \oplus W$ for some $W \subseteq [\ISG(t-1)]_t$.
    It follows that $[A_t B_t]_{t-1} \in \ISG(t-1)$.
    Now the three relevant images are
    $\ance_\tau^t(A_t) = M_t \cdot \ance_\tau^{t-1}([M_t^\dagger A_t ]_{t-1})$,
    $\ance_\tau^t(B_t) = M_t \cdot \ance_\tau^{t-1}([M_t^\dagger B_t ]_{t-1})$, and
    $\ance_\tau^{t}(A_t B_t) = \ance_\tau^{t-1}([A_t B_t]_{t-1})$,
    which are consistent with linearity.\footnote{This part might appear
    to use something special about~$\FF_2$, but does not.
    Over $p$-dimensinal qudits, one must take the commutation value~$c \in \ZZ_p$ into account 
    and consider $M^{\pm c} A$.}

\end{itemize}
\end{proof}

\Cref{thm:anceIsLinear}
says that 
$\ance_{\tau} : \bigoplus_{t \ge \tau} \ISG(t) \to \calE$ with signs ignored
is $\FF_2$-linear.
The kernel of this linear map reveals interdependence among detectors.
For example, the ancestry of a persistent stabilizer~$S$ 
($[S]_r \in \ISG(r)$ for all $r \ge \tau$)
is the stabilizer~$[S]_\tau$ itself.
Hence, two detectors measuring the same persistent stabilizer at different time steps
have the empty ancestor $\one \in \ISG(\tau)$ at~$\tau$.
In this case, the two measurements
must reveal the same eigenvalue in the absence of any errors.
A spacetime-local generating set of the kernel, if exists, 
is useful since the spacetime location of a syndrome bit is more directly 
associated with the spacetime location of a component of an error.
For known Floquet codes~\cite{Hastings2021,Haah2022,Davydova2022}
that are based on topological codes,
there exists such a spacetime-local generating set for the kernel.

\begin{remark}\label{rem:LogicalFailure}
    To probe logical failures, we bring an independent reference system~$R$
    that is at time~$0$ maximally entangled with a code system~$C$.
    The measurements are made on qubits of~$C$ only.
    Some entanglement will break down by measurements on~$C$,
    but if the dynamical code encodes $k$ logical qubits,
    there must be $k$ Bell pairs
    between~$C$ and~$R$.
    At any moment~$t$,
    the entire system~$CR$ is a pure Pauli stabilizer state,
    and every logical operator~$L = L^C$ of~$\ISG(t)$ 
    is a tensor factor
    \footnote{
    This observation was also stated in~\cite{Fuente2024}, where they note that a logical flow from an input stabilizer code to an output stabilizer code defines a logically entangled state between the two.
    } 
    of a stabilizer~$L^C U^R$ of~$CR$.
    By considering the ancestry of~$L^C U^R$,
    we generate an element~$D(L)^{CR} \in \calE^{CR}$ by~\eqref{eq:detectorAsError},
    and the commutation relation between an error and~$D(L)^{CR}$ 
    will tell us if the insertion of the error flips the eigenvalue of~$L^C U^R$.
    Since the error of interest is on~$C$ but not on~$R$,
    we may truncate~$D(L)^{CR}$ to have an element~$D(L)$ on the spacetime of~$C$,
    and use it as a probe for the failure of the logical operator~$L$.
\end{remark}

\begin{remark}\label{rem:LogicalFailureProbeExtensionQuestion}
    Note that it is not clear how to extend~$D(L)$ beyond the time~$t$
    that was used to construct~$D(L)$ in~\cref{rem:LogicalFailure}.
    At minimum, one has to assume that the logical operator~$L$ 
    is not going to be measured at a future time.
    In~\cite{Bombin2024,Fuente2024,Xu2025},
    this problem does not exist because there is always a future temporal boundary,
    which leads to a set of open legs in a corresponding ZX-diagram.
    In our case, we do not have any future boundary
    and a logical failure will be determined by
    finding an instantaneous logical operator at a future time step
    that is equivalent to an undetectable error.
\end{remark}

We record a simple observation that the ``CSS-ness'' of the measurement operators
implies the ``CSS-ness'' of the detectability.

\begin{proposition}\label{thm:CSSmeasurements}
    Suppose that every measurement operator of a dynamical code 
    is either an $X$-type Pauli operator or $Z$-type.
    Then, a spacetime error $E = E_X E_Z$ is undetectable
    if and only if both $E_X$, the $X$-part of~$E$,
    and $E_Z$, the $Z$-part of~$E$, are undetectable.
\end{proposition}

\begin{proof}
    It suffices to show that the element~$D$ of~\eqref{eq:detectorAsError}
    is either $X$-type or $Z$-type.
    Let $M_t$ be a deterministic measurement and suppose it is $Z$-type.
    If $D$ had a nontrivial $X$-component, 
    there must be the greatest time step~$\tau < t$
    and some $Z$-type element~$S \in \ISG(\tau)$
    such that $M_\tau$ is $X$-type, 
    and $[S]_{\tau-1} \notin \ISG(\tau-1)$, 
    but $[M_\tau^\dagger S]_{\tau-1} \in \ISG(\tau-1)$.
    Since all measurements are either $X$- or $Z$-type,
    $\ISG(\tau-1)$ is a direct sum of its $X$-type and $Z$-type subgroups.
    The $Z$-part of~$[M_\tau^\dagger S]_{\tau-1}$ is $[S]_{\tau-1} \in \ISG(\tau-1)$,
    which contradicts the hypothesis that $\tau$ existed.
\end{proof}

\section{Initial and steady stages of dynamical codes}

We will later consider errors that happens during the steady stages of dynamical codes.
Here we give definitions for initial and steady stages of dynamical codes,
and remark on how long initialization step can take for Floquet codes.
Fu and Gottesman~\cite[\S 5]{Fu2024} have shown optimal bounds on the initialization time;
our exposition is perhaps easier to understand.
A simple result (\cref{thm:steadyISG}) in this section will be importantly used
in a later argument that every Floquet code is bounded-inference.

Every Pauli measurement does not decrease the rank of~$\ISG$.
Therefore, given a measurement schedule defined on $n$ qubits, 
there exists a limit
\begin{equation}
    n - k = \lim_{t \to \infty} \rank \ISG(t) \le n,
\end{equation}
from which we read off the number~$k$ of the encoded qubits in the dynamical code.

\begin{definition}
    A time step~$t$ is said to be in the {\bf steady} stage
    if $\rank \ISG(t) = \rank \ISG(t')$ for all~$t' > t$.
    Otherwise, $t$ is in the {\bf initial} stage.
    The {\bf initialization time} is the least~$t \ge 0$ in the steady stage.
    We will often speak of \emph{an error in the steady stage} 
    to mean that the error is supported on a time interval~$[t_0,t_1]$ 
    where $t_0$ is in the steady stage. 
\end{definition}

\begin{proposition}\label{thm:steady}
    If $t$ is in the steady stage and
    if a measurement operator~$M$ at time~$t+1$ is not in~$\barISG(t)$,
    then there exists $S \in \ISG(t)$ that anticommutes with~$M$.
\end{proposition}

\begin{proof}
    Otherwise, $\ISG(t+1) \supseteq \langle \pm M, \ISG(t) \rangle$ whose rank is higher.
\end{proof}

\begin{proposition}\label{thm:steadyISG}
    Let $P$ be the period and $T$ the initialization time of a Floquet code.
    Then, 
    $[\barISG(t)]_{t+P} \subseteq \barISG(t+P)$ for any~$t \ge 0$
    where the equality holds if and only if $t \ge T$.
\end{proposition}

\begin{proof}
    If we show that $\barISG(t) \subseteq \barISG(t+P)$ for any~$t \ge 0$,
    then the equality condition follows 
    by the definition of the initialization time.

    For the set inclusion, 
    we may assume that there is one measurement at a time without loss of generality.
    If $t = 0$, there is nothing to show.
    Let $S \in \ISG(t)$ with $t \ge 1$.
    \Cref{thm:StabilizerFromAncestry} gives a time sequence of operators~$P_1,\ldots,P_t = S$
    where $A_j^\dagger P_j = [P_{j-1}]_j$ commutes with the measurement operator~$M_j$ at time~$j$ 
    for all~$j \in [1,t]$.
    Due to the periodicity,
    we have that $[A_j^\dagger P_j]_{j+P}$ commutes with~$M_{j+P}$
    and hence $[P_t]_{t+P}$ is a member of~$\barISG(t+P)$.
    Therefore, $[S]_{t+P} \in \barISG(t+P)$.
\end{proof}

\begin{corollary}\label{thm:initializationBound}
    For a Floquet code on $n$ qubits with period $P$,
    the initialization time $T$ satisfies
    \begin{equation}
        T \le nP.
    \end{equation}
\end{corollary}

\begin{proof}
    By \cref{thm:steadyISG}, if $t < T$ is in the initial stage,
    then $[\barISG(t)]_{t+P} \subsetneq \barISG(t+P)$,
    which implies $\rank \ISG(t) < \rank \ISG(t+P)$.
    Since $\rank \ISG(t) \le n$ for all $t \ge 0$,
    the rank can increase at most $n$ times.
    Therefore, the initialization time satisfies $T \le nP$.
\end{proof}

\begin{proposition}
    For each even number~$n$ 
    there exists a Floquet code of period~$4$ on $n$ qubits
    whose initialization time is~$T = 2n$.
\end{proposition}

While this particular code is not useful for quantum error correction 
since there will not be any logical qubit,
this demonstrates that the initialization time can scale linearly with system size,
establishing that $T = \Theta(nP)$ in the worst case.
This bound was shown in~\cite[\S 5]{Fu2024}.

\begin{proof}
    Index the qubits by $j \in \{1,2,\ldots,n\}$ and define the measurement schedule:
    \begin{itemize}
        \item At $t = 1 \bmod 4$: measure $Z_j Z_{j+1}$ for $j \in \{1,3,5,\ldots,n-1\}$.
        \item At $t = 2 \bmod 4$: measure $X_j$ for $j \in \{2,4,6,\ldots,n\}$.
        \item At $t = 3 \bmod 4$: measure $Z_j Z_{j+1}$ for $j \in \{2,4,6,\ldots,n-2\}$.
        \item At $t = 4 \bmod 4$: measure $X_j$ for $j \in \{1,3,5,\ldots,n-1\}$.
    \end{itemize}
    After the first measurement at $t=1$, we have $\rank \ISG(1) = n/2$
    from the $ZZ$ measurements.
    The subsequent $X$ measurements at $t=2$ anticommute with these stabilizers.
    The $Z$ measurements at $t=3$ leave $X_{n-1}$ in the ISG, 
    maintaining $\rank\ISG(1) = \rank\ISG(2) = n/2$.
    However, at $t=4$, the ISG consists of $X_j$ on all odd $j$, plus $X_{n}$, 
    so $\rank\ISG(4) = (n/2) + 1$.
    Continuing this accounting, one can verify that in each subsequent period,
    $\rank \ISG(4m) = n/2 + m$ for $m \ge 1$,
    reaching rank~$n$ after $n/2$ periods.
\end{proof}

\section{Benign errors and equivalent errors}

In this section, we elaborate on a seemingly small subset of spacetime errors,
that we call benign,
whose action on the history of density matrices is trivial.
We will see that benign errors and only benign errors 
may define a meaningful equivalence relation on undetectable errors
in view of their logical action.

\subsection{Benign errors}\label{sec:BenignErrors}

\begin{definition}\label{def:BenignErrors}
    A {\bf sandwiching} error $E_{t+1}E_t$ is one that is 
    supported on two consecutive time steps~$t$ and~$t+1$
    such that~$[E_{t+1}]_t  = E_{t}^\dag$ and~$E_{t+1}$ commutes with every measurement operator at time~$t+1$.
    A {\bf vacuous} error is one that equals the measurement operator at a time~$t$. 
    A {\bf benign} error is any finite product of vacuous and sandwiching errors.
    We will refer to either a vacuous or a sandwiching error as a \textbf{benign generator}.
\end{definition}

\noindent
Note that 
while a vacuous error $E_t$ could also be interpreted to be a measurement operator~$[M_{t+1}]_t$, 
such an error is the product of benign generators:
\begin{equation}
    [M_{t+1}]_t = ([M_{t+1}]_t [M_{t+1}^\dagger]_{t+1}) M_{t+1}.
\end{equation}

The action by a vacuous error is the identity 
since it is a multiplication by an operator on its eigenstate
and density matrices are invariant under scalar conjugations.
The action by a sandwiching error is not necessarily the identity action 
on the chain of all density matrices (\cref{rem:excitingBenignError}),
but the second of the sandwiching error cancels the first,
and no measurement outcome is affected.
By~\cref{thm:ErrorAction}, we have
\begin{proposition}\label{thm:BenignIsUndetectable}
    Every benign error is undetectable.
\end{proposition}

\begin{remark}\label{rem:DP2023}
    At the end of~\S5 of~\cite{Delfosse2023},
    essentially the same set of errors as our benign errors is mentioned.
    The authors of~\cite{Delfosse2023} assert
    that these exhaust all undetectable errors ``with trivial effect.''
    Note that the set of all benign errors supported on a time interval~$I$
    can be properly larger than
    the set generated by benign generators on~$I$.
    (See also~\cref{thm:BenignInAWindow}.)
    For example, an instantaneous stabilizer at time~$t$ may have deep ancestry,
    in which case the stabilizer cannot be written as a product of benign generators
    supported on that one time slice~$t$;
    a benign error that is a product of benign generators supported in the future
    may not commute with some elements of~$\ISG(t)$,
    in which case it cannot be a product of benign generators supported in the past and present.
\end{remark}

\begin{remark}\label{rem:excitingBenignError}
    A benign error can make the physical state at some point in time
    different from that predicted by measurement outcomes.
    For example, consider a period-two measurement dynamics on one qubit
    where we measure $X$ at time~1 and then~$Z$ at time~2.
    The insertion of~$Z_2$ immediately after the $Z$ measurement
    is benign (where the subscript is the time), but a sandwiching error $Z_1 Z_2$ 
    makes a benign error $Z_1$ where the underlying state is an eigenstate of~$X$.
\end{remark}

Note that the ancestry of a stabilizer is always a finite product of vacuous errors.
When we recursively find ancestors of a stabilizer,
the placement of $M_t^\dagger S_t $ to the preceding time step 
($[M_t^\dagger S_t]_{t-1}$)
is the multiplication by a sandwiching error.
The commutation requirement is fulfilled by~\cref{thm:StabilizerFromAncestry}.
Hence, an ancestry~$\ance(S_t)$ can always be multiplied 
by a finite product of sandwiching errors to result in~$S_t$.
Explicitly, 
if $\ance_\tau^t(S_t) = \prod_{i=\tau}^t M_i$, 
then we have a telescopic product 
\begin{equation}
    S_t = \ance_\tau^t(S_t) \cdot 
        \prod_{j=\tau+1}^{t} 
        \left(
            \left(\prod_{i = \tau}^j [M_i]_{j} \right)
            \left(\prod_{i = \tau}^j [M_i]_{j-1}^\dagger \right)
        \right) \, .
\end{equation}
As an example, if $\ance_\tau^t(S_t) = M_{t-2} M_{t-1} M_t$, then 
\begin{align}
    S_t 
    &=
    M_{t} M_{t-1} M_{t-2} \cdot  
    \left( \left[M_{t-2}\right]_{t-1} M_{t-2}^\dagger \right) 
    \left( \left[M_{t-1} M_{t-2} \right]_{t} \left[M_{t-1} M_{t-2} \right]_{t-1}^\dagger \right)
    \, .
\end{align}
This implies

\begin{proposition}\label{thm:ISGisBenign}
    For any~$t \ge 0$, every element of~$\ISG(t)$ is benign.
    In fact, it is a product of vacuous errors and sandwiching errors,
    each supported on~$[0,t]$.
    Conversely, if a product~$E$ of benign generators, each supported on~$[0,t]$,
    is supported on one time step~$t$, then $E \in \barISG(t)$.
\end{proposition}

\begin{lemma}\label{thm:BenignLogicalIsStabilizer}
    Any benign error~$(\ldots,\one_{\tau+1}, E_\tau, \one_{\tau-1}, \ldots)$ 
    supported on one time step~$\tau$ in the steady stage
    is an element of~$\barISG(\tau)$, 
    or else, $E_\tau$ anticommutes with some element of~$\ISG(\tau)$.
\end{lemma}

\noindent
This is roughly the converse of~\cref{thm:ISGisBenign}.
This is however false without the assumption of~$\tau$ being in the steady stage.
Suppose that a measurement of an operator~$M_{\tau+1}$ at time~$\tau+1$ 
makes~$\rank \ISG(\tau+1) > \rank \ISG(\tau)$.
This means that $M_{\tau+1}$ commutes with all elements in~$\ISG(\tau)$ but lies outside~$\ISG(\tau)$.
The vacuous error~$M_{\tau+1}$ at~$\tau+1$ can be commuted back to time step~$\tau$,
which is a multiplication by one vacuous error and one sandwiching error,
and becomes a nontrivial logical operator of~$\ISG(\tau)$.
Nonetheless, we will see that the residual impact of a benign error 
always disappears as the dynamics evolves.

\begin{proof}
    We may serialize all measurements~$\{M_t\}$.

    Suppose that $E_\tau$ commutes with every elements of~$\ISG(\tau)$,
    as otherwise the claim is trivially true.
    The error~$E_\tau$ is a product of vacuous and sandwiching errors from the past, present, and future.
    Those from the past and the present can only generate an element of~$\ISG(\tau)$ by~\cref{thm:ISGisBenign}, 
    so we may assume that $E_\tau$ is a product of vacuous errors $\{V_t \in \calE \,:\, \tau < t \le T \}$ 
    and sandwiching errors $\{W_{t,t-1} \in \calE \,:\, \tau < t \le T\}$ from the future. 
    Put $\calB = \{V_t, W_{t,t-1}\,:\, \tau < t \le T \}$.

    We use induction in~$T-\tau$.
    If $T = \tau$, then $\calB = \emptyset$ and we have proven the claim.
    Suppose $T > \tau$.    
    The only error in~$\calB$ that has support on~$\tau$ is the sandwiching error~$W_{\tau+1,\tau}$.
    The product~$W_{\tau+1,\tau}E_\tau$ is then supported on $[\tau+1,T]$.
    This means that $W_{\tau+1,\tau} = E_{\tau+1} E_\tau^\dagger$ where $E_{\tau+1} = [E_\tau]_{\tau+1}$
    and $[M_{\tau+1}]$ commutes with~$[E_\tau]$.
    The next group $\ISG(\tau+1)$ is generated by a subset of~$\ISG(\tau)$ and $\{ M_{\tau+1} \}$,
    both of which element-wise commute with~$E_{\tau+1}$.
    Since $E_{\tau + 1}$ is now a product of benign generators from~$[\tau+1,T]$,
    the induction hypothesis implies that $\pm E_{\tau+1} \in \ISG(\tau+1)$.
    If $\pm [M_{\tau+1}]_\tau \in \ISG(\tau)$, then $\ISG(\tau) = [\ISG(\tau+1)]_{\tau} \ni E_\tau$.
    If $\pm [M_{\tau+1}]_\tau \notin \ISG(\tau)$, 
    then \cref{thm:steady} supplies~$S_\tau \in \ISG(\tau)$ 
    that anticommutes with~$[M_{\tau+1}]_\tau$.
    Now, $E = [E_\tau]$ commutes with~$M = [M_{\tau+1}]$, $S = [S_\tau]$,
    and all the logical operators (trivial or not) of~$[\ISG(\tau+1)]$.
    It follows that $E$ belongs to the intersection of~$[\ISG(\tau)]$ and $[\ISG(\tau+1)]$.
\end{proof}

\subsection{Equivalent errors}

\begin{definition}
    Two spacetime Pauli errors are {\bf equivalent} 
    if their difference is benign.    
\end{definition}

\noindent
Tautologically, every benign error is equivalent to no error.
\Cref{thm:ISGisBenign} says that every element of~$\ISG(t)$ 
for any~$t$ is equivalent to no error.

\begin{corollary}\label{thm:equivalentLogicals}
    If $\tau$ is in the steady stage,
    equivalent logical operators of~$\ISG(\tau)$ differ by~$\ISG(\tau)$.
\end{corollary}
\begin{proof}
    The difference between two logical operators is assumed to be benign
    and commutes with every element of~$\ISG(t)$.
    \Cref{thm:BenignLogicalIsStabilizer} proves the claim.
\end{proof}

\noindent
In the context of static Pauli stabilizer codes,
equivalent logical operators are \emph{defined} in this way.
Here, we have defined the equivalence by our own notion of benign spacetime errors,
and hence the statement is more than a definition.
\Cref{thm:equivalentLogicals} justifies that our definition of equivalence 
is a generalization of that of static Pauli stabilizer codes.
Note that $\tau$ being in the steady stage is an essential assumption
as remarked below~\cref{thm:BenignLogicalIsStabilizer},
which is always satisfied for any static Pauli stabilizer code,
for which we measure every stabilizer every time.

\begin{proposition}\label{thm:benignDoesNothing}
    Equivalent errors $E,E'$ induce the same physical state eventually.
    Specifically,
    if $E'E^\dagger$ is a product of benign generators, each supported on~$[0,t]$,
    then the physical states at time~$t$ obtained by inserting~$E$ and $E'$
    are the same for any valid history of measurement outcomes.
\end{proposition}
\begin{proof}
    \Cref{thm:ErrorAction} says that
    the chain of outcome-recorded density matrices
    upon insertion of~$E'$ can be obtained from that of~$E$ by inserting a benign error~$E'E^\dagger$.
    Clearly, the action of a benign error changes the chain of density matrices
    at only finitely many time steps,
    and the possible changes are only possible in the support of benign generators
    that constitute~$E'E^\dagger$.
\end{proof}

Moreover, the action of an undetectable error supported on the past
can always be reproduced by an error at the present:

\begin{proposition}\label{thm:PushingUndetectable}
    Let $E$ be an undetectable error supported on a time interval~$I = [t_0,t_1]$
    in the steady stage.
    Then, for any time interval~$I' = [t'_0, t'_1]$ with $t'_0 \ge t_0$ and $t'_1 \ge t_1$,
    there exists an undetectable error~$E'$ supported on~$I'$ equivalent to~$E$.
    Moreover, $E'E^\dagger$ is a product of benign generators supported on~$[t_0,t_0']$ and some element of~$\ISG(t_0)$.
    In particular, $E'E^\dagger$ is a product of benign generators 
    supported on~$[0,t_0']$.
\end{proposition}

\begin{proof}
    We may assume that all measurements are serialized.
    It suffices to show the claim for $t'_0 = t_0 + 1$ and $t'_1 = t_1$.
    Let $M_{t_0+1}$ be a measurement operator at time~$t_0 + 1$.
    
    (i) Suppose that $[M_{t_0 + 1}]_{t_0} \in \barISG(t_0)$.
    Then, the measurement of~$M_{t_0+1}$ is a detector,
    and since $E$ is undetectable, 
    the error component~$E_{t_0}$ and the measurement operator~$M_{t_0+1}$ 
    must commute with each other.
    By a sandwiching error $[E_{t_0}]_{t_0+1}E_{t_0}^\dagger$, 
    we can push~$E_{t_0}$ to the next time step.

    (ii) Suppose that $[M_{t_0+1}]_{t_0} \notin \barISG(t_0)$.
    By~\cref{thm:steady},
    there exists $S_{t_0} \in \ISG(t_0)$ that anticommutes with~$[M_{t_0+1}]_{t_0}$.
    If $[M_{t_0+1}]_{t_0}$ commutes with $E_{t_0}$ we use a sandwiching error 
    $[E_{t_0}]_{t_0+1}E_{t_0}^\dagger$ to push~$E_{t_0}$ to the next time step.
    Otherwise, we know that $S_{t_0} E_{t_0}$ can be pushed to the next time step.
    By~\cref{thm:ISGisBenign}, $S_{t_0}$ is benign and is a product of benign generators,
    each of which is supported on~$[0,t_0]$ 
    and hence the push-up is done by a benign error from the past.
\end{proof}

\section{The code distance}

In this section, 
we finally obtain a full characterization of undetectable errors
in what we call bounded-inference dynamical codes,
and will arrive at a well-motivated definition 
of the code distance of dynamical code.
The class of bounded-inference dynamical codes 
includes all Floquet codes,
and is defined by a mild property 
that we believe must be satisfied by all useful dynamical codes.
We will conclude with some facts on Floquet codes,
with which one can write an algorithm for computing the code distance of Floquet codes.
Our definition of the code distance involves the infinite time axis,
and the results at the end of this section guarantee 
that some finite calculation suffices.

\begin{definition}
    The {\bf inference window width} at time~$t$ for a dynamical code
    is the least number~$\mu(t) \in [0,t-1]$
    such that every element~$S \in \ISG(t)$ 
    is determined by measurements supported on
    time interval~$[t-\mu(t),t]$;
    that is, $S$ is a product of benign generators supported on~$[t-\mu(t),t]$.
    A {\bf bounded-inference} dynamical code is one in which $\mu = \sup_t \mu(t) < \infty$.
\end{definition}

Now we use the periodicity of a Floquet code.

\begin{proposition}\label{thm:InferenceWindow}
    Every Floquet code is bounded-inference with an inference window width~$\mu \le T + P -2$
    where $T \ge 1$ is the initialization time and $P \ge 1$ is the period.
\end{proposition}

A static Pauli stabilizer code may be promoted to a Floquet code by measuring 
all stabilizer generators in parallel at every step, so that $P=1$. 
If the $n-k$ generators are measured at $t=1$, then $\rank\ISG(1)=n-k$ and $T=1$. 
The bound $\mu\leq T+P-2=0$ is therefore saturated: each $S\in\ISG(t)$ is inferred 
from the single-round window $[t,t]$, reflecting that all stabilizers are remeasured every step. 
In contrast, the bound can be loose --- 
for the honeycomb code~\cite{Hastings2021}, $T=4$ and $P=3$ while $\mu=3$.

\begin{proof}
    Let $t$ be a time,
    and let $f \in [T,T+P-1]$ be an integer such that $f = t \bmod P$.
    Every element of~$\ISG(f)$ has ancestry supported on~$[1,f]$
    and therefore is inferred by the measurements on the same time interval.
    Since $[\barISG(f)]_t = \barISG(t)$ by~\cref{thm:steadyISG},
    every element of~$\ISG(t)$ is also inferred by the measurements 
    in the time window~$[t - f + 1, t]$.
\end{proof}

The bounded-inference property will be important because of

\begin{lemma}\label{thm:Stabilization}
    Let $E_\tau$ be an undetectable error supported on one time step~$\tau$
    in a bounded-inference dynamical code with a uniform inference window width~$\mu$.
    The physical state of the system at time~$\tau+\mu+1$ with $E_\tau$ inserted 
    is in the code space of~$\ISG(\tau+\mu+1)$.
\end{lemma}

\begin{proof}
    All the eigenvalues of the stabilizers in~$\ISG(\tau+\mu+1)$ 
    and those in later time steps 
    are determined by the outcomes measured after the insertion of~$E_\tau$.
    These eigenvalues are valid and independent of whether~$E_\tau$ is inserted
    since~$E_\tau$ is undetectable.
    This implies that the state is in the correct code space predicted 
    by the measurement outcomes.
\end{proof}

\begin{theorem}\label{thm:NonbenignUndetectedLogical}
    In any bounded-inference dynamical code ({\it e.g.}, a Floquet code),
    an undetectable error in the steady stage
    is always equivalent to an error~$L_t$ supported on one time step~$t$
    such that $L_t$ commutes with every element of~$\ISG(t)$.
\end{theorem}

In fact, as we will see in the proof,
if $E$ is an undetectable error supported on a time interval~$[a,b]$,
then the operator~$L_t$ can be found by ``pushing''~$E$ into the time step~$b + \mu+1$ 
by applying~\cref{thm:PushingUndetectable}.

\begin{proof}
    Let $\mu$ be an inference window width.
    \cref{thm:InferenceWindow} says that a Floquet code has a bounded inference window.
    By~\cref{thm:PushingUndetectable} an undetectable error is equivalent to
    another undetectable error~$E_\tau$ supported on one time step~$\tau$.
    \cref{thm:Stabilization} says that the physical state at time $t=\tau+\mu+1$
    with $E_\tau$ inserted is in the instantaneous code space.
    By~\cref{thm:PushingUndetectable} again,
    we multiply~$E_\tau$ by a benign error consisting of 
    benign generators in the past of~$t$
    to obtain an equivalent error~$L_t$.
    The action of~$L_t$ on a pristine state gives the same physical state
    as the one with $E_\tau$ inserted by~\cref{thm:benignDoesNothing},
    and hence the action of a Pauli operator~$L_t$ results in the code space.
    Therefore, $L_t$ must commute with all the stabilizers at time~$t$.
\end{proof}

\begin{corollary}\label{thm:LogicalFailureNonbenign}
    Let $E$ be an undetectable error on a bounded-inference dynamical code 
    (\it{e.g.}, a Floquet code) in the steady stage. 
    The error~$E$ is correctable if and only if $E$ is benign.
\end{corollary}

\noindent 
Put verbosely,
if $E$ is benign, then the physical state eventually becomes pristine
with the code's measurement dynamics;
if $E$ is not benign, no decoder may correct~$E$ 
and the encoded qubits are unrecoverably damaged.

\begin{proof}
    \Cref{thm:NonbenignUndetectedLogical} says that 
    $E$ is equivalent to a logical operator~$L$, trivial or not, of~$\ISG(t)$ at some time~$t$.
    If $L$ is a stabilizer, then $L$ is benign, and so is~$E$.
    In this case there is no damage to the logical qubit.
    If $L$ is not a stabilizer, then $L$ cannot be corrected,
    and \cref{thm:BenignLogicalIsStabilizer} says that $L$ is not benign.
\end{proof}

\begin{remark}
    \Cref{thm:NonbenignUndetectedLogical,thm:LogicalFailureNonbenign}
    remain true
    for slightly more general dynamical codes 
    where the inference window width~$\mu(t)$
    satisfies
    \begin{equation}
        \lim_{t \to \infty} t - \mu(t) = \infty.
    \end{equation}
\end{remark}

We finally arrive at a desired characterization of the code distance of Floquet codes.
Let a {\bf weight} be a function $\calE \to \ZZ_{\ge 0}$.
One must use a well motivated function for a weight,
reflecting the underlying noise of qubits.
If every location in spacetime suffers from independent stochastic noise,
it makes sense to define the weight to be the number of nonidentity factors of a spacetime error,
which is usually assumed and we, too, use below.

\begin{definition}
    The {\bf code distance} of a bounded-inference dynamical code
    is the minimal weight of a nonbenign undetectable error~$E\in \calE$ 
    inserted into the steady stage.
\end{definition}

\noindent
This is to be compared with a conventional notion:
the code distance of a Pauli stabilizer code is 
the minimal weight of a nonstabilizer undetectable error.

\begin{proposition}
    Promote a static Pauli stabilizer code to a Floquet code
    by measuring all stabilizer generators every time step.
    Then, the code distance of this Floquet code is equal to
    the conventional code distance of the Pauli stabilizer code.
\end{proposition}

\begin{proof}
    Since the code distance of the Floquet code
    cannot exceed the code distance~$d$ of the instantaneous
    Pauli stabilizer code,
    it suffices to consider undetectable spacetime errors of
    weight at most~$d$.
    Since we infer all stabilizers every time step,
    every time slice of an undetectable spacetime error 
    must be undetectable on its own;
    otherwise, 
    the earliest nonlogical slice must flip some of the succeeding stabilizer measurements.
    Since the total weight is at most~$d$,
    the weight in each slice is at most~$d$
    and hence each slice is a stabilizer unless the weight is~$d$.
    Every stabilizer is benign by~\cref{thm:ISGisBenign}.
\end{proof}

\begin{corollary}
    It is no easier to compute the code distance of a Floquet code
    than to compute the conventional code distance of a Pauli stabilizer code.
\end{corollary}

\begin{theorem}\label{thm:dOver2Correctable}
    If $d$ is the code distance of a Floquet code,
    any spacetime error of weight less than~$d/2$ occurred during the steady stage
    can be corrected.
\end{theorem}

\begin{proof}
    Such an error~$E$ of weight less than~$d/2$ is either detectable or benign.
    If it is benign, no correction is needed;
    once we go sufficiently far into the future past the error,
    the physical state is the same as if no error has occurred.
    If it is detectable, then we have to find a finite time window
    in which the error can be supported.
    Given such a time window, 
    we can solve an inhomogeneous $\FF_2$-linear equation to find some error~$C$
    that reproduces the same syndrome and has weight less than~$d/2$.
    The combination~$EC^\dagger$ has weight less than~$d$ and is undetectable,
    and hence is benign by~\cref{thm:LogicalFailureNonbenign}.

    It remains to determine a sufficiently large but finite time window that supports~$E$
    (or its relevant factors).
    To this end, we redefine detectors as follows.
    Bring~$\mu$ from~\cref{thm:InferenceWindow}.
    Our detector at~$t$ is redefined to be the difference between
    the outcome of the detector at~$t$ and 
    the prediction made by the outcomes in the time window~$[t-\mu,t]$.\footnote{We have discussed a similar choice of a basis for syndrome above in the context of ancestors of stabilizers.}

    Under this redefinition of detectors,
    there must exist a nonidentity component of~$E$ within 
    the $\mu$-neighborhood of an unhappy detector along the time axis;
    otherwise, the detector cannot see the error and hence cannot be unhappy.
    If $E$ stretches farther along the time axis than $\mu d/2$ from unhappy detectors
    of the greatest or least time coordinate,
    then there must be a time interval of length~$> \mu$ that lacks any error factor.
    In that case, the factors of~$E$ that sits far away from the unhappy detectors
    must be undetectable on their own, and hence benign.
    In conclusion, it suffices to look for~$C$ in the time window~$[t_0 - \mu d/2, t_1+ \mu d/2]$
    where $t_0$ is the time coordinate of the earliest unhappy detector and $t_1$ is that of the latest.
\end{proof}

\begin{remark}\label{rem:measurementErrors}
    We have not discussed any measurement errors so far,
    but they can be accounted for by a conjugating Pauli error that anticommutes with 
    the measurement operator if all measurement operators at a time step have nonoverlapping support.
    When the weight means the number of nonidentity Pauli factors in an error,
    this means that a measurement error corresponds to a Pauli error of weight~$2$.
    If a measurement dynamics contains overlapping measurement operators at a time step,
    the representation of measurement errors by conjugating Pauli
    may not be appropriate since the conjugating Pauli error may
    correspond to two or more measurement outcome flips.
    In that case, one can introduce fictitious time steps
    to make the measurement operators nonoverlapping.
\end{remark}

\begin{remark}\label{rem:ComputingFloquetCodeDistance}
    Computing the code distance of a Pauli stabilizer code is in general inefficient~\cite{Iyer2013},
    and for Floquet code it is only more complex.
    \footnote{
        This is a necessary consequence 
        since a Floquet code can express all syndrome measurement gadgets along with data qubits
        and furthermore erases the distinction between data and ancilla qubits.
    }
    A new problem is that the time axis is infinite.
    However, thanks to periodicity, some finite computation is enough.
    Suppose that an undetectable error~$E$ is a product~$A B$ 
    where $A$ is supported on a time interval~$[a_0,a_1]$ and $B$ on~$[b_0,b_1]$.
    If $a_1 < b_0 - \mu$ where $\mu$ is a uniform bound on inference time window
    (\cref{thm:InferenceWindow}),
    then each of~$A$ and~$B$ must be undetectable.
    Thus, in the computation of the code distance of a Floquet code,
    it suffices to consider errors~$E$ that do not allow such decomposition,
    implying that the time support of~$E$ must be ``connected.''
    Let $d_0$ be the minimum of the code distances of all~$\ISG$'s in the steady stage.
    We see that it suffices to consider all undetected errors supported on~$[T, T + d_0 \mu]$
    to determine the code distance of a Floquet code.
    In practice, one must further restrict the time window to ease the search
    by, for example,
    (i) considering subsystem codes defined by two neighboring~$\ISG$'s to reduce~$d_0$,
    (ii) optimizing the inference window of detectors,
    and
    (iii) lower bounding the weight on each time slice 
    for an error to be undetectable from immediately following detectors.
    To test if an undetectable error is benign,
    one can use either~\cref{rem:LogicalFailure} or~\cref{thm:BenignInAWindow}.
\end{remark}

\begin{proposition}\label{thm:BenignInAWindow}
    For any bounded-inference dynamical code with inference window width~$\mu$,
    all benign errors supported on a time window~$[a,b]$ in the steady stage
    are generated by benign generators supported on~$[a - \mu, b + \mu]$.
\end{proposition}

See \cref{sec:LadderCode} for an example.

\begin{proof}
    Let $E$ be a benign error supported on~$[a,b]$.
    We know that $E$ is undetectable by~\cref{thm:BenignIsUndetectable},
    and hence is equivalent to an error~$F_b$ supported on one time step~$b$
    such that $F_b E^\dagger$ is a product of an element of~$\barISG(a)$
    and some benign generators supported on~$[a,b]$
    by~\cref{thm:PushingUndetectable}.
    Since the dynamical code is assumed to be bounded-inference,
    $F_b E^\dagger$ is a product of benign generators supported on~$[a-\mu,b]$.
    We may push~$F_b$ further onto the time step~$b+\mu$
    to obtain~$G_{b+\mu}$ that commutes with every element 
    of $\ISG(b+\mu)$ by~\cref{thm:NonbenignUndetectedLogical}.
    $G_{b+\mu} F_b^\dagger$ is a product of benign generators on~$[b-\mu,b+\mu]$.
    Since $G_{b+\mu}$ is benign, 
    $G_{b+\mu}$ is an element of $\barISG(b+\mu)$
    by~\cref{thm:BenignLogicalIsStabilizer}
    and hence is a product of benign generators on~$[b,b+\mu]$.
    We have found a decomposition $E = (E F_b^\dagger) (F_b G_{b+\mu}^\dagger) G_{b+\mu}$
    of~$E$ in terms of benign generators on~$[a-\mu,b+\mu]$.
\end{proof}

\section{Examples}

Here, we examine three Floquet codes.
The first is a toy code from~\cite{Hastings2021}
that is called a ladder code. 
We include this mainly for pedagogical purposes.

The second example is a planar version of the honeycomb code
from~\cite{Vuillot2021}.
This example was noted to have a constant-weight spacetime error that causes logical failure,
but an explicit calculation seemed complicated.
In view of previous results~\cite{Hastings2021,Fuente2024}, 
it would seem necessary that one has to first calculate all detectors,
find a suitable undetectable spacetime error, 
and then show that the error induces nontrivial logical action.
Using our results, the calculation for this example becomes simple:
we just have to find an equivalent spacetime error that is a logical (trivial or not)
operator of some instantaneous stabilizer group
by ``pushing'' the error towards the future after multiplying by suitable benign errors.
By~\cref{thm:equivalentLogicals,thm:NonbenignUndetectedLogical},
this pushing procedure will reveal a unique logical operator (trivial or not)
if and only if the original error was undetectable.
This way, we can bypass the calculation of any detectors 
and the logical action of the error becomes evident.

The third example is the Floquet Bacon--Shor code of~\cite{Alam2024}.
An upper bound on the spacetime code distance was shown,
and numerical evidence was given that the upper bound was sharp.
By examining the requirement for an undetectable error 
to be equivalent to another error supported arbitrarily far in the future,
we will give a simple proof that the reported upper bound is indeed the spacetime code distance 
of the Floquet Bacon--Shor code.

\subsection{Ladder code}\label{sec:LadderCode}

\begin{figure}[ht]
  \centering
  \includegraphics[width=0.75\textwidth]{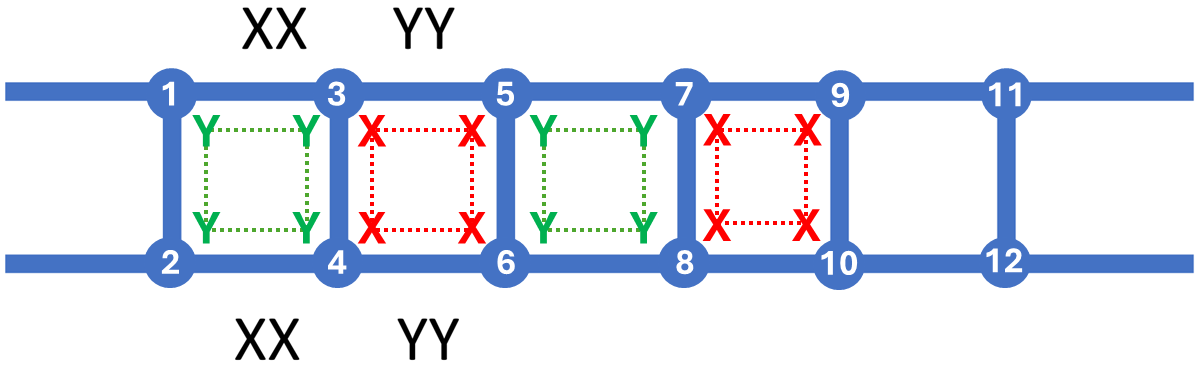}
  \caption{Ladder code for $m=3$.  Qubits are associated to vertices, enumerated as shown.  Measurement operators on vertical legs are $ZZ$ on the two qubits.  
          Measurement operators on horizontal legs are alternately $XX$ or $YY$; 
          a couple are shown. 
          Persistent stabilizers of the parent subsystem code are weight-4 plaquettes of type~$Y$ (green) or~$X$ (red) 
          which are not measured directly, 
          but rather inferred through the measurement schedule.
  }
  \label{fig:ladder}
\end{figure}

There are $n=4m$ physical qubits placed on vertices 
with the periodic boundary conditions as shown in~\cref{fig:ladder}; 
the ends of the legs are connected. 
The faces are labeled in an alternating $0$ and $1$ pattern. 
The period $P=4$ measurement schedule is as follows.
\begin{itemize}
    \item ($t = 1 \bmod 4$) Measure two-qubit $ZZ$ on the rungs.
    \item ($t = 2 \bmod 4$) Measure two-qubit $XX$ on horizontal leg segments straddling $0$-labeled faces.
    \item ($t = 3 \bmod 4$) Measure two-qubit $ZZ$ on the rungs.
    \item ($t = 4 \bmod 4$) Measure two-qubit $YY$ on horizontal leg segments straddling $1$-labeled faces.
\end{itemize}
The code has persistent stabilizers shown in~\cref{fig:ladder}; 
we let $S_X$ and $S_Y$ denote the persistent $X$- and $Y$-type stabilizers, respectively.
With the exception of $i = 1$, 
$S_X$ is inferred following the measurement of~$\mathbf M_i$ for $i = 1,4 \bmod 4$ 
and $S_Y$ is inferred following the measurement of~$\mathbf M_i$ for $i = 2,3 \bmod 4$. 
The inference window is sliding with constant width $\mu = 2$, 
meaning persistent stabilizers are inferred every round (except $i=1$) 
and that the inference is based on measurements taken at time steps $i$ and $i-1$ only. 
When $n = 4m$, we have $|S_Y| = |S_X| = m$.

\cref{thm:BenignInAWindow} says that the set of benign errors
supported on a time interval $[a,b]$ is generated by benign generators
supported on a larger time window $[a-\mu, b+\mu]$.
Let us see why this larger time window is necessary.
Let $a = 3$ and $b = 4$. 
Consider a spacetime error
\begin{equation}
    E = E_a E_b= \left[ X^1 X^2 X^3 X^4 \right]_{a} \left[ X^1 X^2 X^3 X^4 \right]_b \, .
\end{equation}
Note that $E$ is \emph{not} a sandwiching error
since $X^1 X^2 X^3 X^4$ does not commute with~$Y^3 Y^5 \in \mathbf M_4$.
Still, $E$ is benign; observe
\begin{itemize}
    \item[(i)] $E_a \in \ISG(a)$; 
    indeed, it is the product of a persistent stabilizer $Y^1 Y^2 Y^3 Y^4$ 
    and the measurement operators $Z^1 Z^2$ and $Z^3 Z^4$ measured at time $t=a$; and

    \item[(ii)]  $E_b \notin \ISG(b)$, but it commutes past $\mathbf M_{b+1}$. 
    Since $[E_b]_{b+1} \in \barISG(b+1)$, 
    it is then reabsorbed and has no lasting effect on the history of physical states.
\end{itemize}
By~\cref{thm:ISGisBenign}, the error~$E$ decomposes 
into a product of benign generators. 
Ignoring signs, we have 
\begin{align}
    E_{a=3} &= \left[ Z^1 Z^2 \right]_a \left[ Z^3 Z^4\right]_a \left[Y^1 Y^2 Y^3 Y^4\right]_a \, ,\\
    \ance_{a=3}^{b+1=5}([E_b]_{b+1}) 
    &= \left[ Z^1 Z^2 \right]_{b+1} \left[Z^3 Z^4 \right]_{b+1} 
       \left[ Y^1 Y^2 Y^3 Y^4 \right]_{a}. \nonumber
\end{align}
Hence, although $E$ is a benign error supported on~$[a,b]$, 
it is a product of vacuous errors at~$t=b+1$ and~$t=a$ 
with sandwiching errors in~$[a,b+1]$.

We next demonstrate the ``pushing'' algorithm that appears in the proof of~\cref{thm:PushingUndetectable}.
Let $\tau = 4j$ for some $j \ge 1$ corresponds to a time step 
for which the $YY$ measurements $\mathbf M_4$ are implemented. 
Consider a spacetime error of weight~$7$ given by
\begin{equation}\label{eqn:HeavyLadderError}
    E_{[\tau,\tau+3]} = 
    \left[ Z^6 X^7 Y^8 \right]_{\tau+2}
    \left[ Z^4 Z^8 \right]_{\tau+1} 
    \left[ X^3 X^6 \right]_{\tau} 
\end{equation}
It may not be obvious how to go about determining 
whether the error is undetectable, and if so, 
whether it induces a nontrivial logical action.
Let us try to push $E_{[\tau,\tau + 3]}$ into the future.
\begin{itemize}
    \item Between~$\tau$ and $\tau+1$: 
    $X^3 X^6$ anticommutes with $Z^3 Z^4, Z^5 Z^6 \in \mathbf M_1$; 
    accordingly, we insert the vacuous error $V_\tau = [Y^3 Y^5]_\tau$ so that
    \[
        V_\tau \cdot [X^3 X^6]_\tau = [Y^3 Y^5]_\tau [X^3 X^6]_\tau = [Z^3 Y^5 X^6]_\tau
    \]
    now commutes with all measurement operators in $M_1$. 
    Since $V_t \cdot [X^3 X^6]_\tau$ commutes with $M_1$, we may insert the sandwiching error 
    $W_{\tau+1} = \left[Z^3 Y^5 X^6 \right]_\tau \left[Z^3 Y^5 X^6 \right]_{\tau+1}$
    to cancel the time-$\tau$ component.

    \item Between~$\tau + 1$ and~$\tau+2$:
    the original component~$[Z^4 Z^8]_{\tau+1}$ is modified
    to $[Z^3 Z^4 Y^5 X^6 Z^8]_{\tau+1}$.
    We insert a vacuous error 
    $V_{\tau+1} = [Z^3 Z^4]_{\tau+1} [Z^7 Z^8]_{\tau+1}$ 
    to ensure commutativity with~$\mathbf M_2$:
    \[
        V_{\tau+1}\cdot W_{\tau+1} \cdot V_t \cdot \left[ Z^4 Z^8 \right]_{\tau+1} \left[ X^3 X^6 \right]_{\tau} = \left[ Y^5 X^6 Z^7 \right]_{\tau+1}.
    \]    
    We insert a sandwiching error
    $ W_{\tau+2} = \left[Y^5 X^6 Z^7\right]_{\tau+1} \left[Y^5 X^6 Z^7 \right]_{\tau+2}$.
    \item At~$\tau+2$:
    finally, we obtain
    \[
        E'_{[\tau+2]} = W_{\tau+2} \cdot  V_{\tau+1} \cdot W_{\tau+1} \cdot V_t \cdot E_{[\tau,\tau+2]} = \left[Y^5 Y^6 Y^7 Y^8 \right]_{\tau+2} \in \left[ S_Y \right]_{\tau+2}.
    \]
\end{itemize}
Thus, we have shown that the spacetime error~$E_{[\tau,\tau+2]}$ 
is equivalent to a persistent stabilizer at time~$\tau+2$.
Therefore, $E_{[\tau,\tau+2]}$ is benign and hence is an undetectable (correctable) error.

\subsection{Planar honeycomb code by Vuillot}

Vuillot's code~\cite{Vuillot2021} is shown in~\cref{fig:vuillot-honeycomb}, 
with physical qubits associated with the vertices.
It is a planar variant of the honeycomb code~\cite{Hastings2021}.
The measurement schedule has period $P=3$, 
cycling through measurements of $XX, ZZ, YY$, in order. 
At each time step, the bulk protocol is identical to~\cite{Hastings2021} 
while the boundary introduces a pair of defects, 
denoted by circles on opposite ends of the hexagonal patch.
These defects, called \emph{corners}, 
are single-qubit measurements in the same Pauli basis associated with the given time step.
As in~\cite{Hastings2021}, 
the instantaneous codes remain Clifford-equivalent to a surface code, with distance $\Theta(\sqrt{n})$.
\Cref{fig:vuillot-honeycomb} shows the Clifford-equivalent surface code at each time step; 
in the figure, 
the distance of each instantaneous code is~5.
Hence, the spacetime distance~$\delta$ is bounded as $2 \leq \delta \leq \Theta(\sqrt{n})$.

\begin{figure}[ht]
    \centering
    \includegraphics[width=0.6\linewidth]{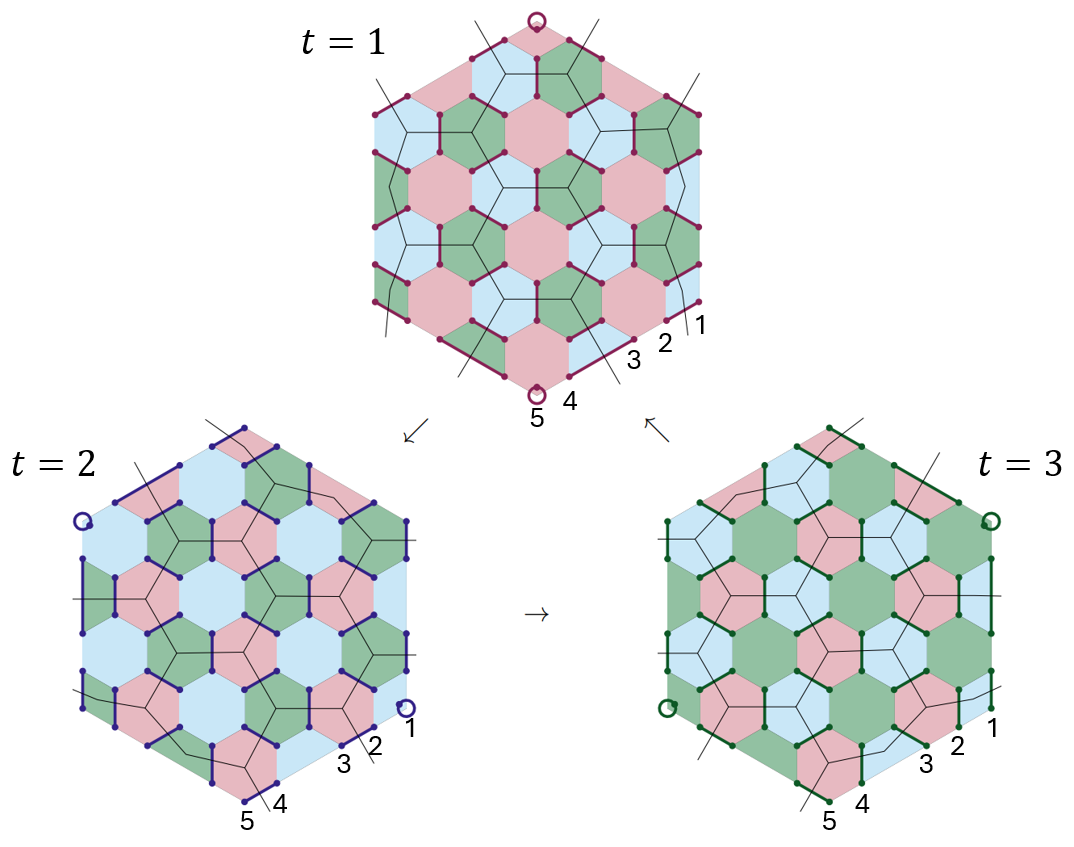}
    \caption{Measurement schedule of planar-boundary honeycomb variant. Measurement operators at each time step are two-body measurements, with a pair of single-qubit measurements at opposite corners. Red, blue, and green measurement operators are supported in $X$, $Z$, and $Y$, respectively. Thin black lines represent the equivalent hexagonal surface code describing the state at each time step. Illustration provided by C. Vuillot and used with permission.
    }
    \label{fig:vuillot-honeycomb}
\end{figure}
 
We are going to prove that a spacetime error
\[
    E = \left[ X^1 \right]_1 \left[ Y^1 \right]_3
\]
is undetectable and induces logical failure, so $\delta = 2$,
underscoring that large instantaneous code distances 
are insufficient for a large spacetime code distance.
Here, the qubits on the hexagonal patch are labeled as shown 
in~\cref{fig:vuillot-honeycomb}.
The proof of the claim follows a similar calculation as in the ladder code.
We will find equivalent errors following the procedure in the proof of~\cref{thm:PushingUndetectable},
until we see that an equilvalent error is a logical operator~$L$ of some~$\ISG(t)$.
Then, since a logical operator of an~$\ISG$ is always undetectable,
and by~\cref{thm:BenignIsUndetectable} benign errors are undetectable,
it follows that $E$ itself is undetectable.
\Cref{thm:benignDoesNothing} implies that the action on the physical state by~$E$ 
is the same as that by~$L$, and hence $E$ causes a logical failure.

Below, we will write shorthand $X^{1,2,3,5}$ and $Y^{2,3}$, {\it etc.}, 
in place of~$X^1 X^2 X^3 X^5$ and $Y^2 Y^3$.

\begin{itemize}
    \item $\mathbf{t=1}$: 
    the error $\left[ X^1 \right]_1$ is introduced to the system. 
    $X^1$ does not commute with the $Z$-measurement at time $t=2$; 
    following~\cref{thm:PushingUndetectable}, 
    we multiply $\left[X^1\right]_1$ 
    by a vacuous error $\left[ X^{1,2,3,4,5} \right]_1 \in \ISG(1)$, 
    which is a product of measurement operators $X^{1,2}, X^{3,4}, X^{5}$ at $t=1$. 
    The resulting error is given by
    \[
        \left[X^1\right]_1 \left[ X^{1,2,3,4,5} \right]_1 = \left[ X^{2,3,4,5} \right]_1
    \]
    which commutes with the $Z$-type measurement operators at time $t=2$.

    \item $\mathbf{t=2}$: 
    the error $\left[ X^{2,3,4,5} \right]_1$ commutes with the $Z$-measurements 
    and commutes forward to become $\left[ X^{2,3,4,5} \right]_2$ 
    via a sandwiching error~$W$; 
    explicitly,
    \[
        W_2 = \left[ X^{2,3,4,5} \right]_1 \left[ X^{2,3,4,5} \right]_2.
    \]
    The error $\left[ X^{2,3,4,5} \right]_2$ does not commute with the $Y$-measurements at~$t=3$, 
    so we again apply~\cref{thm:PushingUndetectable} 
    and multiply by a vacuous error $\left[ Z^{2,3,4,5} \right]_2$, 
    the product of measurement operators $Z^{2,3}$, $Z^{4,5}$ at~$t=2$.
    The resulting error is given by
    \[
        \left[ X^{2,3,4,5} \right]_2 \left[ Z^{2,3,4,5} \right]_2 = \left[ Y^{2,3,4,5} \right]_2.
    \]

    \item $\mathbf{t=2}$: 
    the error $\left[ Y^{2,3,4,5} \right]_2$ commutes with the $Y$-measurements at time $t=3$,
    so we may introduce a sandwiching error 
    $W_3 = \left[ Y^{2,3,4,5} \right]_2 \left[ Y^{2,3,4,5} \right]_3$ 
    to commute the error forward to become $\left[ Y^{2,3,4,5} \right]_3$.
    At this point, the error $\left[ Y^1 \right]_3$ is introduced,
    resulting in $\left[ Y^{1,2,3,4,5} \right]_3$.
\end{itemize}

Thus, $E = \left[ X^1 \right]_1 \left[ Y^1 \right]_3$ is equivalent to 
$E' = \left[ Y^{1,2,3,4,5} \right]_3$, which is a nontrivial logical operator of~$\ISG(3)$.
This completes the proof of the claim that~$\delta = 2$.

\begin{figure}[ht]
    \centering
    \includegraphics[width=0.8\linewidth]{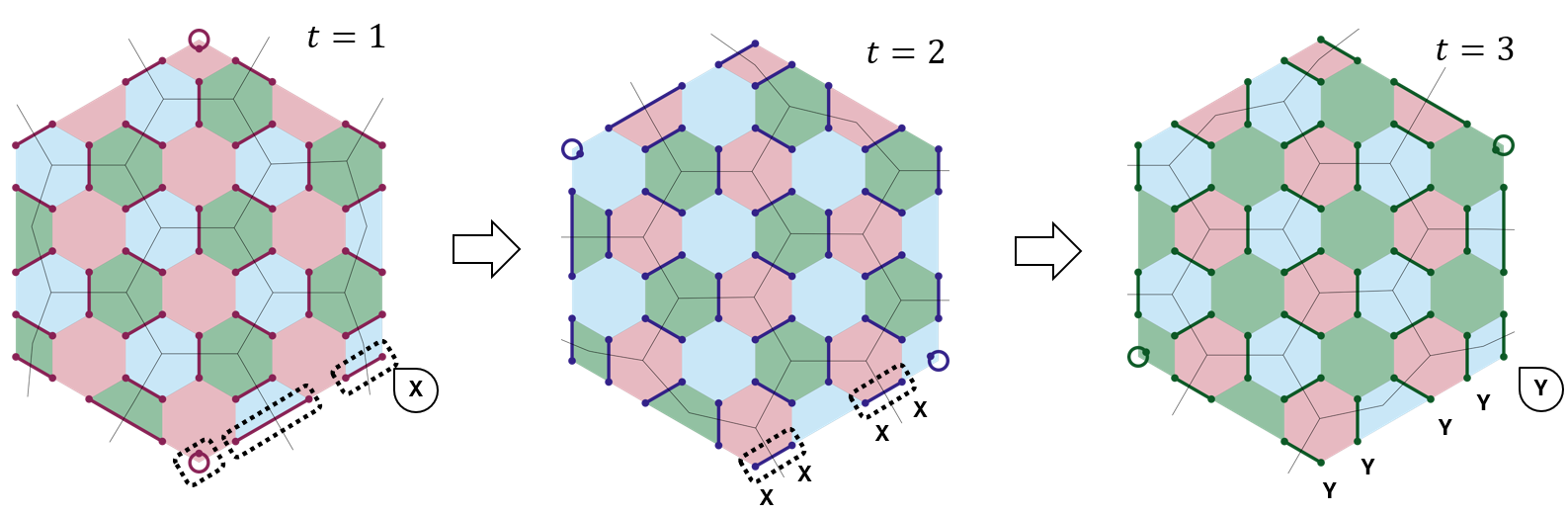}
    \caption{Compression of weight-2 spacetime error supported on time window $\{1,2,3\}$ 
        to a single-time error at time $t=3$, 
        at which point the error is a logical of $\ISG(3)$.
        Past supports of the spacetime error are commuted forward 
        by multiplying the error by measurement operators in the dashed boxes, 
        following the procedure of~\cref{thm:PushingUndetectable}.
        Teardrop arrows indicate the introduction of spacetime error. 
        Illustration provided by C. Vuillot and used with permission.
    }
    \label{fig:vuillot-spacetime-err}
\end{figure}

\subsection{Floquet Bacon--Shor code}

Alam and Rieffel~\cite{Alam2024} construct a Floquet variant of 
the Bacon--Shor code~\cite{Bacon2006}, a prototype of subsystem codes. 
While stabilizer codes directly measure the stabilizer elements, 
subsystem codes decompose the stabilizer elements over a basis of so-called \emph{gauge operators}, 
which span a generally nonablelian subgroup of the Pauli group. 
It is shown that the Floquet variant of the Bacon--Shor code can support $k$ dynamical logical qubits 
(in addition to one static logical qubit) 
through a 4-periodic measurement schedule with $k$ ``gauge defects''. 
Here we only discuss the case with $k=1$.
The code is defined on a $d \times d$ grid for odd $d = 2c+1$; 
physical qubits are associated to vertices.
The measurement schedule is shown in~\cref{fig:FBS-schedule}.

\begin{figure}[ht]
    \centering
    \includegraphics[width=1\linewidth]{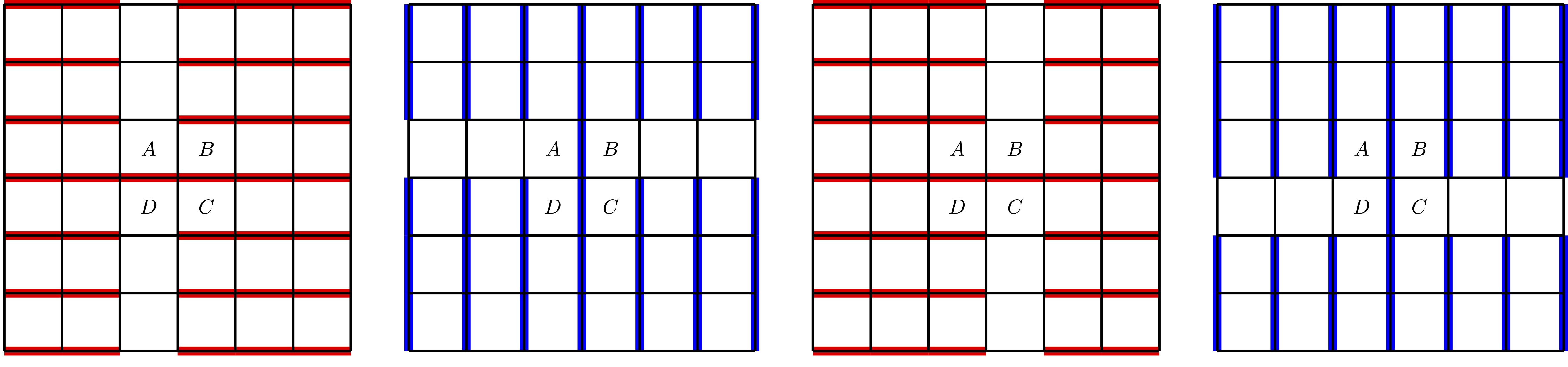}
    \caption{4-periodic measurement schedule of the Floquet Bacon--Shor code. 
        At times $t=1,3 \bmod 4$, 
        $XX$ is measured on the horizontal line segments. 
        When $t=2,4 \bmod 4$, 
        $ZZ$ is measured on the vertical line segments. We denote each set of measurements by $M_i$, $i \in [4]$.
        Reproduced from Figure~3 of Alam and Rieffel~\cite{Alam2024}, licensed under CC-BY 4.0.
    }
    \label{fig:FBS-schedule}
\end{figure}

\cite{Alam2024} proved that the Floquet Bacon-Shor code has spacetime distance bounded above by $2c$, 
which was validated in simulation.
Moreover, they conjecture the bound is tight. 
We prove this conjecture by applying \cref{thm:NonbenignUndetectedLogical},
demonstrating its use in the distance analysis of Floquet codes.

\begin{proposition}\label{thm:FBSDistance}
    Let $d \geq 3$ be odd.
    The spacetime distance of the Floquet Bacon--Shor code encoding two logical qubits 
    (one static and one dynamic) on a $d \times d$ square lattice is $\delta = d-1$.
\end{proposition}

We first show the following lemma. 
Let $d = 2c+1$ and index the vertices of the $d \times d$ square lattice by $[-c,c]^2$.
We say that an error $O = (O_{t_0}, \ldots, O_{t_1})$ is of $X$-type 
if $O_i \in \lbrace I, X \rbrace^{\otimes d^2}$ for all $i \in [t_0, t_1]$.

\begin{lemma}\label{thm:FBSSimplification}
    Let $\Tilde{E} = (\Tilde{E}_i)_{i \in [\alpha,\beta]}$ 
    be an $X$-type spacetime error supported on odd time steps
    such that each $\Tilde{E}_i \neq I$ 
    is neither a static logical operator nor a persistent stabilizer of the Bacon--Shor code 
    on the $d \times d$ grid.
    Then, there exists an equivalent $X$-type spacetime error 
    $E = (E_i)_{i \in [a,b]}$ of weight at most that of~$\tilde E$ such that
    \begin{itemize}
        \item[(i)] 
        every nonidentity $E_i$ anticommutes with some element of~$\mathbf M_{i+1}$; and
        \item[(ii)] $E_i$ is supported on the set $\mathcal{V}$ of far left, right vertices 
        of the grid, 
        except for the midpoint of the right edge  (see~\cref{fig:FBS-aux}).
    \end{itemize}
\end{lemma}

\begin{proof}
For the first claim, observe that any $\Tilde{E}_i$ commuting with $M_{i+1}$ 
may be commuted forward by the sandwiching error $W_{i+1} = [\Tilde{E}_i ]_{i+1} \Tilde{E}_i$, 
which can only lower the spacetime weight of~$\Tilde{E}$. 
After $\tau \leq 4$ such steps, $\Tilde{E}_{i + \tau}$ must anticommute with $\mathbf M_{\tau + 1}$; 
otherwise, $\Tilde{E}_i$ is a static logical/stabilizer.
The second claim follows from observing that multiplication by $X$-type vacuous errors 
can translate the support of $\Tilde{E}_i$ to the far left, right vertices of the lattice 
without increasing the spacetime weight.
\end{proof}

\begin{figure}[ht]
    \centering
    \includegraphics[width=1\linewidth]{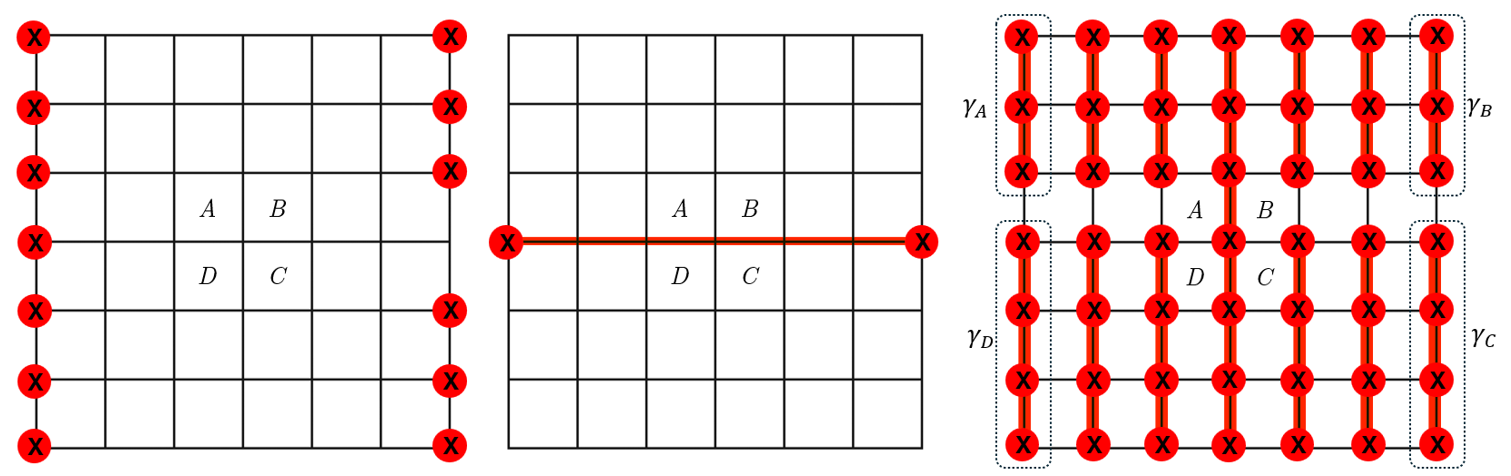}
    \caption{
        \textbf{Left:} up to vacuous errors from $M_1$, $X$-type errors are generated by single-qubit $X$ operators supported on the far left, right vertices of the lattice. Any error component lying on the $x$-axis is translated to the left. \textbf{Center:} the only element of $\langle M_1 \rangle$ supported on the far left, right vertices. \textbf{Right:} the basis of all $X$-type errors commuting with $M_2$; the far left, right elements are labeled as shown.
        Modified from Figure 3 of Alam and Rieffel~\cite{Alam2024}, licensed under CC-BY 4.0.
    }
    \label{fig:FBS-aux}
\end{figure}

\begin{lemma}\label{thm:FBSSTuck}
    Let $t \geq 5$ be an odd time step. 
    Let $E_t$ be a one-time $X$-type error supported on~$\mathcal V$ of weight~$< 2c$
    such that $E_t$ anticommutes with some element of~$\mathbf M_{t+1}$.
    Then, $S E_t$ anticommutes with some element of~$M_{t+1}$ for all $S \in \ISG(t)$.
\end{lemma}

\begin{proof}
    Assume otherwise.
    Without loss of generality, we may fix $t=1 \bmod 4$.
    \cite{Alam2024} shows that any $X$-type element~$S \in \ISG(t)$ 
    is a product of persistent Bacon--Shor stabilizers and some elements of~$\mathbf M_1$.
    Since persistent Bacon--Shor stabilizers commute with every measurement,
    we may assume that $S \in \langle \mathbf M_1 \rangle$.
    On the other hand, the only operator that commutes with every element of~$\mathbf M_2$
    is a product of ``vertical strings'' depicted in the right pane of~\cref{fig:FBS-aux}.
    So, $S E_t$ must be some product of these vertical strings.
    It is clear by inspection that we can further modify~$S E_t$ by~$\mathbf M_1$
    such that $S' E_t$ is commuting with all of~$\mathbf M_2$ and $S' E_t$ is supported 
    on the left, right edges of the grid.
    Then, since $E_t$ is supported on the left, right edges,
    $S' = (S' E_t) (E_t)$ is supported on the left, right edges,
    but there is only one possible choice of $S'$ as shown in the middle pane of~\cref{fig:FBS-aux}.
    It follows that $E_t$ must contain a factor $S' \gamma_C \gamma_D$ of weight~$2c$.
\end{proof}

The proof of~\cref{thm:FBSDistance} is now straightforward.

\begin{proof}[Proof of~\cref{thm:FBSDistance}]    
    Let $E=(E_a, \ldots, E_b)$ be an undetectable spacetime error satisfying the conditions of~\cref{thm:FBSSimplification};
    otherwise, since the instantaneous code distance (of the Bacon--Shor code) is $2c+1$,
    there is nothing to prove.
    In particular, $E_a$ anticommutes with $M_{a+1}$. 
    By~\cref{thm:PushingUndetectable}, 
    there must exist an element $S \in \ISG(a)$ such that $SE_a$ commutes with $\mathbf M_{a+1}$. By~\cref{thm:FBSSTuck}, the weight of the first nonidentity time slice of~$E$ has weight~$\ge 2c$. 
    Combining this with the upper bound~$2c$ from~\cite{Alam2024} completes the proof.
\end{proof}

\section{Discussion}

We have derived a dichotomy of undetectable errors in bounded-inference dynamical codes
that benign errors and only benign errors are correctable.
The generating set for the benign errors is
as simple as one would guess:
the measurement operator inserted immediately after a measurement,
called a vacuous error,
and a pair of the same Pauli operators inserted on consecutive time steps
which would cancel each other
but which conjugate measurements that are undisturbed by the insertion,
called a sandwich error.
Just like Pauli stabilizers leave a static code state invariant,
these benign errors leave all but finitely many density matrices invariant
in the history of the physical state of the code.
The dichotomy is proved by showing that 
the action of any undetectable error on the history of the code
can always be realized by an essentially unique logical operator (trivial or not)
of an instantaneous stabilizer code sufficiently far in the future.
This is the best one can hope for;
some undetectable error may not commute with all instantaneous stabilizers.
To bound how far into the future one must go,
we have used the bounded-inference property,
which is enjoyed by all Floquet (time periodic) codes.
Based on the dichotomy, we have defined the code distance of Floquet codes.

Importantly, our results do not assume any time boundary conditions,
and are derived only based on the defining data of a dynamical or Floquet code.
Even if a dynamical code is defined on a finite time window,
for example by destructive single-qubit measurements in the end,
the code distance in our sense reproduces 
what is typically called the ``effective distance'' of a code
that is measured in numerical performance tests 
in the presence of circuit level noise.
Our code distance thus enables intrinsic comparison between Floquet codes,
generalizing the role of the code distance in static Pauli stabilizer codes.

Below are further comments and open problems.

\begin{remark}\label{rem:QubitReplacement}
    We have assumed that there is a fixed set of qubits from the beginning of time,
    on which a dynamical code is defined.
    This is reasonable if a system of qubits is fabricated once and for all,
    but does not obviously cover situations 
    in which physical qubits are discarded and replenished frequently.
    However, our setting can handle the latter case
    provided that any single qubit is measured in a Pauli basis before it is discarded
    and any fresh qubit is initialized in Pauli basis before it interacts with any other qubit,
    which we believe is a reasonable assumption.
    The discard can be modeled by 
    the completely depolarizing channel~$\calD: \rho \mapsto \frac 1 2 \one$ on a qubit~$q$.
    It is no harm for the purpose of analysis
    to declare that a fresh qubit that compensates a discarded qubit
    is the same qubit~$q$ that just went through~$\calD$.
    Overall, the qubit~$q$ is acted on
    by a single-qubit Pauli measurement~$M$, and then~$\calD$, 
    and then another single-qubit Pauli measurement~$M'$.
    If $M$ and $M'$ do not commute,
    then by benign errors we can insert any single-qubit Pauli error at the location of~$\calD$.
    This effectively removes the role of~$\calD$,
    and it suffices to consider a circuit without~$\calD$ but with~$M$ and $M'$.
    If $M$ and $M'$ commute,
    then we may insert a noncommuting measurement for the purpose of analysis in between~$M$ and~$M'$,
    and we may remove~$\calD$ from further consideration.
\end{remark}

\begin{remark}
    A logical operator in the steady stage 
    is an undetectable error supported on one time step.
    Every time step we have an equivalent logical operator 
    for each logical operator we start with.
    These evolved logical operators set the frame of reference,
    based on which we can read logical qubit.
    So, it is \emph{not} meaningful to say that some Floquet code implements
    some logical Clifford transformation every cycle.
    For example, the honeycomb code without boundary~\cite{Hastings2021}
    may appear to implement the logical Clifford corresponding 
    to $e \leftrightarrow m$,
    but this transformation is never used for any computational tasks.
    A meaningful statement can be made when a Floquet code 
    or its temporary modification such as insertion 
    of a transversal logical unitary gate
    is \emph{compared} to another code block 
    and the two code blocks may interact.
\end{remark}

\begin{problem}
    Periodicity in time resembles translation invariance in space.
    Spatially translation-invariant codes can be compactly represented
    by polynomial matrices~\cite{Haah2013},
    and a similar representation is possible
    for a space-translation-invariant Floquet codes.
    In the static translation-invariant case,
    the excitation map (the map from errors to syndromes)
    is the adjoint of the stabilizer map.
    It is unknown however
    what an analogous excitation map is
    for spacetime translation-invariant Floquet codes.
    The ancestry would be a starting point for this question, 
    and it remains to find an explicit generating set of the kernel of~$\ance_0$
    on the domain of all deterministic measurements.
\end{problem}

\begin{problem}
    Many statements in this paper cease to be true 
    in the initial stage of a dynamical code.
    The initial stage is perhaps the most important for the question of 
    preparation of logical eigenstates.
    What logical states can be fault-tolerantly initialized
    without too heavy modification of the initialization step?
    More concretely, 
    we can ask what logical states can be prepared fault-tolerantly
    by running the measurement schedule of a Floquet code 
    starting with a product state.
    In a similar vein, what logical operators can be measured fault-tolerantly
    by single-qubit measurements especially if a Floquet code is not CSS?
\end{problem}

\bibliography{fd-refs}

\newcommand{\etalchar}[1]{$^{#1}$}
\begin{thebibliography}{dlFOTT{\etalchar{+}}25}

\bibitem[AR24]{Alam2024}
M.~Sohaib Alam and Eleanor Rieffel.
\newblock Dynamical logical qubits in the bacon-shor code.
\newblock {\em Physical Review A}, 2024.
\newblock \href {https://arxiv.org/abs/2403.03291} {\path{arXiv:2403.03291}}, \href {https://doi.org/10.1103/nfxv-3dp7} {\path{doi:10.1103/nfxv-3dp7}}.

\bibitem[Bac06]{Bacon2006}
Dave Bacon.
\newblock Operator quantum error‐correcting subsystems for self‐correcting quantum memories.
\newblock {\em Physical Review A}, 73(1):012340, 2006.
\newblock \href {https://arxiv.org/abs/quant-ph/0506023} {\path{arXiv:quant-ph/0506023}}, \href {https://doi.org/10.1103/PhysRevA.73.012340} {\path{doi:10.1103/PhysRevA.73.012340}}.

\bibitem[BLN{\etalchar{+}}24]{Bombin2024}
Hector Bombin, Daniel Litinski, Naomi Nickerson, Fernando Pastawski, and Sam Roberts.
\newblock Unifying flavors of fault tolerance with the {ZX} calculus.
\newblock {\em Quantum}, 8:1379, 2024.
\newblock \href {https://arxiv.org/abs/2303.08829} {\path{arXiv:2303.08829}}, \href {https://doi.org/10.22331/q-2024-06-18-1379} {\path{doi:10.22331/q-2024-06-18-1379}}.

\bibitem[DKLP02]{Dennis2001}
Eric Dennis, Alexei Kitaev, Andrew Landahl, and John Preskill.
\newblock Topological quantum memory.
\newblock {\em J. Math. Phys.}, 43(9):4452--4505, 2002.
\newblock \href {https://arxiv.org/abs/quant-ph/0110143} {\path{arXiv:quant-ph/0110143}}, \href {https://doi.org/10.1063/1.1499754} {\path{doi:10.1063/1.1499754}}.

\bibitem[dlFOTT{\etalchar{+}}25]{Fuente2024}
Julio C.~Magdalena de~la Fuente, Josias Old, Alex Townsend-Teague, Manuel Rispler, Jens Eisert, and Markus Müller.
\newblock Xyz ruby code: Making a case for a three-colored graphical calculus for quantum error correction in spacetime.
\newblock {\em PRX Quantum}, 6(1):010360, March 2025.
\newblock \href {https://arxiv.org/abs/2407.08566} {\path{arXiv:2407.08566}}, \href {https://doi.org/10.1103/prxquantum.6.010360} {\path{doi:10.1103/prxquantum.6.010360}}.

\bibitem[DP23]{Delfosse2023}
Nicolas Delfosse and Adam Paetznick.
\newblock Spacetime codes of clifford circuits.
\newblock 2023.
\newblock \href {https://arxiv.org/abs/2304.05943} {\path{arXiv:2304.05943}}, \href {https://doi.org/10.48550/ARXIV.2304.05943} {\path{doi:10.48550/ARXIV.2304.05943}}.

\bibitem[DTB23]{Davydova2022}
Margarita Davydova, Nathanan Tantivasadakarn, and Shankar Balasubramanian.
\newblock {Floquet} codes without parent subsystem codes.
\newblock {\em PRX Quantum}, 4(2):020341, 2023.
\newblock \href {https://arxiv.org/abs/2210.02468} {\path{arXiv:2210.02468}}, \href {https://doi.org/10.1103/prxquantum.4.020341} {\path{doi:10.1103/prxquantum.4.020341}}.

\bibitem[FDB{\etalchar{+}}23]{Fahimniya2023}
Ali Fahimniya, Hossein Dehghani, Kishor Bharti, Sheryl Mathew, Alicia~J. Kollár, Alexey~V. Gorshkov, and Michael~J. Gullans.
\newblock Fault-tolerant hyperbolic {Floquet} quantum error correcting codes.
\newblock 2023.
\newblock \href {https://arxiv.org/abs/2309.10033} {\path{arXiv:2309.10033}}, \href {https://doi.org/10.48550/ARXIV.2309.10033} {\path{doi:10.48550/ARXIV.2309.10033}}.

\bibitem[FG24]{Fu2024}
Xiaozhen Fu and Daniel Gottesman.
\newblock Error correction in dynamical codes.
\newblock 2024.
\newblock \href {https://arxiv.org/abs/2403.04163} {\path{arXiv:2403.04163}}, \href {https://doi.org/10.48550/ARXIV.2403.04163} {\path{doi:10.48550/ARXIV.2403.04163}}.

\bibitem[GNFB21]{Gidney2021}
Craig Gidney, Michael Newman, Austin Fowler, and Michael Broughton.
\newblock A fault-tolerant honeycomb memory.
\newblock {\em Quantum}, 5:605, 2021.
\newblock \href {https://arxiv.org/abs/2108.10457} {\path{arXiv:2108.10457}}, \href {https://doi.org/10.22331/q-2021-12-20-605} {\path{doi:10.22331/q-2021-12-20-605}}.

\bibitem[GNM22]{Gidney2022}
Craig Gidney, Michael Newman, and Matt McEwen.
\newblock Benchmarking the planar honeycomb code.
\newblock {\em Quantum}, 6:813, 2022.
\newblock \href {https://arxiv.org/abs/2202.11845} {\path{arXiv:2202.11845}}, \href {https://doi.org/10.22331/q-2022-09-21-813} {\path{doi:10.22331/q-2022-09-21-813}}.

\bibitem[Got22]{Gottesman2022}
Daniel Gottesman.
\newblock Opportunities and challenges in fault-tolerant quantum computation.
\newblock 2022.
\newblock \href {https://arxiv.org/abs/2210.15844} {\path{arXiv:2210.15844}}, \href {https://doi.org/10.48550/ARXIV.2210.15844} {\path{doi:10.48550/ARXIV.2210.15844}}.

\bibitem[Haa13]{Haah2013}
Jeongwan Haah.
\newblock Commuting pauli hamiltonians as maps between free modules.
\newblock {\em Commun. Math. Phys.}, 324(2):351--399, October 2013.
\newblock \href {https://arxiv.org/abs/1204.1063} {\path{arXiv:1204.1063}}, \href {https://doi.org/10.1007/s00220-013-1810-2} {\path{doi:10.1007/s00220-013-1810-2}}.

\bibitem[HH21]{Hastings2021}
Matthew~B. Hastings and Jeongwan Haah.
\newblock Dynamically generated logical qubits.
\newblock {\em Quantum}, 5:564, 2021.
\newblock \href {https://arxiv.org/abs/2107.02194} {\path{arXiv:2107.02194}}, \href {https://doi.org/10.22331/q-2021-10-19-564} {\path{doi:10.22331/q-2021-10-19-564}}.

\bibitem[HH22]{Haah2022}
Jeongwan Haah and Matthew~B. Hastings.
\newblock Boundaries for the honeycomb code.
\newblock {\em Quantum}, 6:693, 2022.
\newblock \href {https://arxiv.org/abs/2110.09545} {\path{arXiv:2110.09545}}, \href {https://doi.org/10.22331/q-2022-04-21-693} {\path{doi:10.22331/q-2022-04-21-693}}.

\bibitem[IP15]{Iyer2013}
Pavithran Iyer and David Poulin.
\newblock Hardness of decoding quantum stabilizer codes.
\newblock {\em IEEE Trans. Info. Theory}, 61:5209--5223, 2015.
\newblock \href {https://arxiv.org/abs/1310.3235} {\path{arXiv:1310.3235}}, \href {https://doi.org/10.1109/TIT.2015.2422294} {\path{doi:10.1109/TIT.2015.2422294}}.

\bibitem[MBG23]{McEwen2023}
Matt McEwen, Dave Bacon, and Craig Gidney.
\newblock Relaxing hardware requirements for surface code circuits using time-dynamics.
\newblock {\em Quantum}, 7:1172, 2023.
\newblock \href {https://arxiv.org/abs/2302.02192} {\path{arXiv:2302.02192}}, \href {https://doi.org/10.22331/q-2023-11-07-1172} {\path{doi:10.22331/q-2023-11-07-1172}}.

\bibitem[PKD{\etalchar{+}}23]{Paetznick2022}
Adam Paetznick, Christina Knapp, Nicolas Delfosse, Bela Bauer, Jeongwan Haah, Matthew~B. Hastings, and Marcus~P. da~Silva.
\newblock Performance of planar {Floquet} codes with {Majorana}-based qubits.
\newblock {\em PRX Quantum}, 4(1):010310, 2023.
\newblock \href {https://arxiv.org/abs/2202.11829} {\path{arXiv:2202.11829}}, \href {https://doi.org/10.1103/prxquantum.4.010310} {\path{doi:10.1103/prxquantum.4.010310}}.

\bibitem[vdW20]{Wetering2020}
John van~de Wetering.
\newblock Zx-calculus for the working quantum computer scientist.
\newblock 2020.
\newblock \href {https://arxiv.org/abs/2012.13966} {\path{arXiv:2012.13966}}, \href {https://doi.org/10.48550/ARXIV.2012.13966} {\path{doi:10.48550/ARXIV.2012.13966}}.

\bibitem[Vui21]{Vuillot2021}
Christophe Vuillot.
\newblock Planar {Floquet} codes.
\newblock 2021.
\newblock \href {https://arxiv.org/abs/2110.05348} {\path{arXiv:2110.05348}}, \href {https://doi.org/10.48550/ARXIV.2110.05348} {\path{doi:10.48550/ARXIV.2110.05348}}.

\bibitem[XD25]{Xu2025}
Yichen Xu and Arpit Dua.
\newblock Fault-tolerant protocols through spacetime concatenation.
\newblock 2025.
\newblock \href {https://arxiv.org/abs/2504.08918} {\path{arXiv:2504.08918}}, \href {https://doi.org/10.48550/ARXIV.2504.08918} {\path{doi:10.48550/ARXIV.2504.08918}}.

\bibitem[YCC24]{Yan2024}
Bowen Yan, Penghua Chen, and Shawn~X. Cui.
\newblock {Floquet} codes from coupled spin chains.
\newblock 2024.
\newblock \href {https://arxiv.org/abs/2410.18265} {\path{arXiv:2410.18265}}, \href {https://doi.org/10.48550/ARXIV.2410.18265} {\path{doi:10.48550/ARXIV.2410.18265}}.

\bibitem[ZAV23]{Zhang2022}
Zhehao Zhang, David Aasen, and Sagar Vijay.
\newblock The {X}-cube {Floquet} code.
\newblock {\em Phys. Rev. B}, 108:205116, 2023.
\newblock \href {https://arxiv.org/abs/2211.05784} {\path{arXiv:2211.05784}}, \href {https://doi.org/10.1103/PhysRevB.108.205116} {\path{doi:10.1103/PhysRevB.108.205116}}.

\end{thebibliography}
\bibliographystyle{alphaurl}
\end{document}